\documentclass[11pt]{article}

\usepackage{algorithm,amsmath,amssymb,amsthm,boxedminipage,color,url,fullpage,mathtools}
\usepackage{algpseudocode}
\usepackage{bbm}
\usepackage{enumitem}
\usepackage{stmaryrd}
\usepackage{float}
\usepackage[numbers]{natbib} 
\usepackage[pdfstartview=FitH,colorlinks,linkcolor=blue,filecolor=blue,citecolor=blue,urlcolor=blue]{hyperref} 
\usepackage[labelfont=bf]{caption}
\usepackage{aliascnt,cleveref}
\usepackage{graphicx}
\usepackage{makecell}
\usepackage[table]{xcolor}\usepackage{xparse}

\newtheorem{theorem}{Theorem}[section]
\newtheorem{corollary}[theorem]{Corollary}
\newtheorem{fact}[theorem]{Fact}
\newtheorem{lemma}[theorem]{Lemma}
\newtheorem{definition}[theorem]{Definition}
\newtheorem{remark}[theorem]{Remark}

\newtheorem{assumption}[theorem]{Assumption}

\newcommand{\norm}[1]{\left\| #1 \right\|}

\newcommand{\Dminushalf}{U^{-\frac{1}{2}}}

\NewDocumentCommand\muG{gg}{\ensuremath{\mu\IfNoValueTF{#1}{}{_{#1}}\IfNoValueTF{#2}{}{(#2)}}}
\NewDocumentCommand\muGG{g}{\ensuremath{\mu\IfNoValueTF{#1}{}{(#1)}}}

\newcommand{\ter}{\mathcal{T}} 
\newcommand{\muBound}{K}

\newcommand{\Uminushalf}{U^{-\frac{1}{2}}}
\newcommand{\Uhalf}{U^{\frac{1}{2}}}
\newcommand{\normalized}[1]{\Uminushalf #1 \Uminushalf}

\newcommand{\EE}{\mathop{\mathbb E}} 
\newcommand{\RR}{\mathop{\mathbb R}}
\newcommand{\NN}{\mathop{\mathbb N}}

\newcommand{\CC}{\mathop{\mathbb C}}

\newcommand{\1}{\mathbbm{1}}

\newcommand{\ie}{{\it i.e.}}

\newcommand{\etal}{{\it et al.}}

\newcommand*\comp[1]{\overline{#1}}

\DeclareMathOperator{\vol}{\mathbf{vol}}
\DeclareMathOperator{\poly}{\mathbf{poly}}
\DeclareMathOperator{\tr}{Tr}
\DeclareMathOperator{\diag}{\mathbf{diag}}
\newcommand{\defeq}{\mathrel{\mathop:}=}

\newcommand\numeq[1]{\stackrel{\scriptscriptstyle(\mkern-1.5mu#1\mkern-1.5mu)}{=}}
\newcommand\numge[1]{\stackrel{\scriptscriptstyle(\mkern-1.5mu#1\mkern-1.5mu)}{\ge}}
\newcommand\numle[1]{\stackrel{\scriptscriptstyle(\mkern-1.5mu#1\mkern-1.5mu)}{\le}}

\NewDocumentCommand\F{gg}{\ensuremath{F\IfNoValueTF{#1}{}{_{#1}}\IfNoValueTF{#2}{}{(#2)}}}
\NewDocumentCommand\M{gg}{\ensuremath{M\IfNoValueTF{#1}{}{_{#1}}\IfNoValueTF{#2}{}{(#2)}}}
\NewDocumentCommand\N{gg}{\ensuremath{N\IfNoValueTF{#1}{}{_{#1}}\IfNoValueTF{#2}{}{(#2)}}}
\NewDocumentCommand\W{gg}{\ensuremath{W\IfNoValueTF{#1}{}{_{#1}}\IfNoValueTF{#2}{}{(#2)}}}
\NewDocumentCommand\D{gg}{\ensuremath{D\IfNoValueTF{#1}{}{_{#1}}\IfNoValueTF{#2}{}{(#2)}}}
\NewDocumentCommand\A{gg}{\ensuremath{A\IfNoValueTF{#1}{}{_{#1}}\IfNoValueTF{#2}{}{(#2)}}}
\NewDocumentCommand\B{gg}{\ensuremath{B\IfNoValueTF{#1}{}{_{#1}}\IfNoValueTF{#2}{}{(#2)}}}
\NewDocumentCommand\X{gg}{\ensuremath{X\IfNoValueTF{#1}{}{_{#1}}\IfNoValueTF{#2}{}{(#2)}}}
\NewDocumentCommand\Y{gg}{\ensuremath{Y\IfNoValueTF{#1}{}{_{#1}}\IfNoValueTF{#2}{}{(#2)}}}
\NewDocumentCommand\Z{gg}{\ensuremath{Z\IfNoValueTF{#1}{}{_{#1}}\IfNoValueTF{#2}{}{(#2)}}}
\NewDocumentCommand\R{gg}{\ensuremath{R\IfNoValueTF{#1}{}{_{#1}}\IfNoValueTF{#2}{}{(#2)}}}
\NewDocumentCommand\Smat{gg}{\ensuremath{S\IfNoValueTF{#1}{}{_{#1}}\IfNoValueTF{#2}{}{(#2)}}}
\NewDocumentCommand\U{gg}{\ensuremath{U\IfNoValueTF{#1}{}{_{#1}}\IfNoValueTF{#2}{}{(#2)}}}
\NewDocumentCommand\Hmat{gg}{\ensuremath{H\IfNoValueTF{#1}{}{_{#1}}\IfNoValueTF{#2}{}{(#2)}}}
\NewDocumentCommand\Pmat{gg}{\ensuremath{P\IfNoValueTF{#1}{}{_{#1}}\IfNoValueTF{#2}{}{(#2)}}}
\NewDocumentCommand\I{gg}{\ensuremath{I\IfNoValueTF{#1}{}{_{#1}}\IfNoValueTF{#2}{}{(#2)}}}
\NewDocumentCommand\dG{gg}{\ensuremath{\mu\IfNoValueTF{#1}{}{_{#1}}\IfNoValueTF{#2}{}{(#2)}}}
\NewDocumentCommand\Matrix{mm}{\ensuremath{\left(#1\right)(#2)}}
\NewDocumentCommand\MatrixSimple{mm}{\ensuremath{#1(#2)}}
\NewDocumentCommand\dGG{g}{\ensuremath{\mu\IfNoValueTF{#1}{}{(#1)}}}
\NewDocumentCommand\vv{g}{\ensuremath{v\IfNoValueTF{#1}{}{(#1)}}}

\begin{document}
	\title{Expander Decomposition for Non-Uniform Vertex Measures\thanks{This work was supported in part by Israel Science Foundation grant no.\ 1595-19, 1156-23, and the Blavatnik Family Foundation.}
	}
	
	\author{Daniel Agassy\thanks{Tel Aviv University,
danielagassy@mail.tau.ac.il} \and Dani Dorfman\thanks{Max Planck Institute for Informatics, Saarbr\"ucken , Germany,
ddorfman@mpi-inf.mpg.de} 	\and Haim Kaplan\thanks{Tel Aviv University, haimk@tau.ac.il}}
	
	\maketitle
	
	\noindent

\begin{abstract}
A $(\phi,\epsilon)$-expander-decomposition of a graph $G$ (with $n$ vertices and $m$ edges) is a partition of $V$ into clusters $V_1,\ldots,V_k$ with conductance $\Phi(G[V_i]) \ge \phi$, such that there are $O(\epsilon m)$ inter-cluster edges. Such a decomposition plays a crucial role in many graph algorithms. \cite{agassy2022expander}  gave a randomized $\tilde{O}(m)$ time algorithm for computing a $(\phi, \phi\log^2 {n})$-expander decomposition.

\looseness = -1

In this paper we generalize this result for a broader notion of expansion. Let $\muG \in \RR_{\ge 0 }^n$ be a vertex measure. A standard generalization of conductance of a cut $(S,\comp{S})$ is its $\muG$-expansion $\Phi^{\muG}_G(S,\comp{S}) = |E(S,\comp{S})|/\min \{\muG{}{S},\muG{}{\comp{S}}\}$, where $\muG(S) = \sum_{v\in S} \muG{}{v}$. \\
We present a randomized $\tilde{O}(m)$ time algorithm for computing a $(\phi, \phi \log^2 {n}\cdot\frac{\muG{}{V}}{m})$-expander decomposition with respect to $\muG$-expansion.
A substantial portion of the exposition is adapted from~\cite{agassy2022expander}, and this work serves as a convenient reference for generalized expander decomposition.

\end{abstract}

\section{Introduction}
\looseness=-1
The \emph{conductance} of a cut  $(S, \comp{S})$ is $\Phi_{G}(S,\comp{S}) = \frac{|E(S,\comp{S})|}{\min(\vol(S), \vol(\comp{S}))}$, where $\vol(S)$ is the sum of the degrees of the vertices of $S$. The  conductance of a graph $G$ is the smallest  conductance of a cut in $G$. 

A \emph{$(\phi,\epsilon)$-expander decomposition} of a graph $G$ is a partition of the vertices of $G$ into clusters $V_1,\ldots,V_k$, each with conductance $\Phi(G[V_i]) \ge \phi$ such that there are at most $O(\epsilon m)$ inter-cluster edges, 
where $\phi,\epsilon\ge 0$. Consider the problem of computing 
a  $(\phi,\epsilon)$-expander decomposition for a given graph $G$ and $\phi>0$, while minimizing $\epsilon$ as a function of $\phi$.  
It is known that a $(\phi,\epsilon)$-expander decomposition, with $\epsilon = O(\phi\log n)$, always exists and that $\epsilon = \Theta(\phi\log n)$ is best possible~\cite{saranurak2019expander,alev2017graph}. In a recent result, \cite{agassy2022expander} achieved a $(\phi, \phi\log^2 {n})$-expander decomposition in $\tilde{O}(m)$ time. 

The following is an overview on the result of \cite{agassy2022expander} which  builds on~\cite{saranurak2019expander}. The naive way to compute an expander decomposition is to use an approximation algorithm for the sparsest cut problem. Given an $f(n)$-approximation algorithm for the problem of finding a minimum conductance cut, one can get a $(\phi, O(f(n)\cdot\phi\log n))$-expander decomposition algorithm by recursively computing approximate cuts until all components are certified as expanders. The drawback of this approach is a possible linear recursion depth (because of unbalanced cuts - one side has very small volume) which affects the running time. To solve this problem, Saranurak and Wang~\cite{saranurak2019expander} developed a sparse cut algorithm that either returns a \emph{balanced} sparse cut or an unbalanced sparse cut in which the larger side is an expander (and therefore they only need to recur on the smaller side of the cut). Their algorithm is based on the Cut-Matching Game~\cite{khandekar2009graph}. 

{\bf Cut-matching:} 
The cut-matching game, played by 
a \emph{cut player} and a \emph{matching player} to construct an expander, is a well known technique to reduce the approximation task for sparsest cut to a polylogarithmic number of maximum flow problems~\cite{khandekar2009graph}.
 Khandekar \etal~\cite{khandekar2009graph} devised a strategy for the cut player that translates to an $O(\log^2 {n})$ approximation algorithm for the sparsest cut problem. 
 This cut player of~\cite{khandekar2009graph} is robust, in the sense that it is simple and can be flexibly adapted  for other applications~\cite{racke2014computing,saranurak2019expander}.
 A more sophisticated spectral cut player was devised by Orecchia \etal~\cite{orecchia2008partitioning} which yields $O(\log{n})$ approximation. This cut player is less flexible and harder to modify.

By modifying the cut-matching game and the cut player of~\cite{khandekar2009graph}, \cite{saranurak2019expander} achieved a $(\phi, \phi \log^3{n})$-expander decomposition in $\tilde{O}(\frac{m}{\phi})$ time. In a recent result~\cite{agassy2022expander} 
were able by a careful spectral analysis, to combine the techniques of~\cite{saranurak2019expander} and the spectral cut player of~\cite{orecchia2008partitioning} to get a $(\phi, \phi \log^2{n})$-expander decomposition in $\tilde{O}(\frac{m}{\phi})$ time. Using the recent fair-cuts result~\cite{li2022fair}, both expander decomposition algorithms~\cite{saranurak2019expander,agassy2022expander} can be implemented in $\tilde{O}(m)$ time.

A useful generalization of conductance is $\muG$-expansion. Let $\muG \in V\rightarrow \RR_{\ge 0}^n$ be a vertex measure.  The $\muG$-expansion of a cut  $(S, \comp{S})$ is $\Phi^{\muG}_{G}(S,\comp{S}) = \frac{|E(S,\comp{S})|}{\min(\muG{}{S}, \muG{}{\comp{S}})}$, where $\muG{}{S} = \sum_{v\in S}\muG{}{v}$.
If $\min(\muG{}{S}, \muG{}{\comp{S}}) = 0$ then $\Phi^{\muG}_{G}(S,\comp{S})= \infty$.  Note that by setting $\muG{}{v} = \deg_G(v)$, we get conductance. The $\muG$-expansion of $G$ is defined as $\Phi^{\muG}(G) = \min_{S \subseteq V} \Phi^{\muG}_{G}(S,\comp{S})$. We say that $G$ is a $(\phi, \muG)$-expander if $\Phi^{\muG}(G) \ge \phi$.

A \emph{$(\phi,\epsilon,\muG)$-expander decomposition} of a graph $G$ is a partition of the vertices of $G$ into clusters $V_1,\ldots,V_k$, each with $\muG$-expansion $\Phi^{\muG}(G[V_i]) \ge \phi$ such that there are at most $O(\epsilon m)$ inter-cluster edges, where $\phi,\epsilon\ge 0$ and $\muG \in V\rightarrow \RR_{\ge 0}^n$.

{\bf Our Contribution:} 
In this paper we generalize the result of~\cite{agassy2022expander} to the notion of $\muG$-expansion and present an $\tilde{O}(m)$ algorithm for computing a $(\phi, \phi \log^2 {n} \frac{\muG{}{V}}{m})$-expander decomposition with respect to $\muG$-expansion. The generalization is quite straightforward; we derive it by 
modifying a simplified version of~\cite{agassy2022expander} (this simplification is possible due to a recent result on fair-cut~\cite{LNPSsoda13}). As an immediate consequence of our balanced sparse cut procedure, used in the expander decomposition algorithm (see Lemma~\ref{theorem:main-result}), we get an $O(\log{n})$ approximation algorithm to the balanced sparse cut problem with respect to a vertex measure $\muG \in \RR_{\ge 0}^n$.

 The core ideas of this paper are the same as in our previous work~\cite{agassy2022expander}. However, the generalization we pursue here alters the technical development and many internal details substantially. For this reason, rather than layering modifications onto~\cite{agassy2022expander}, we provide a fresh and self-contained account of the result. Several passages are therefore intentionally adapted from~\cite{agassy2022expander} to maintain continuity and clarity of exposition.

{\bf Related work:}     
In a recent paper, Ameranis \etal~\cite{pmlr-v162-orecchia22a} use a generalized notion of $\muG$-expansion, also mentioned in \cite{orecchia2011fast}, where they allowed to consider overlapping partitions of $V$. 
They define a corresponding generalized version of the cut-matching game, and show how to use a cut strategy for this game to get an approximation algorithm for two generalized cut problems.
They claim that one can construct a cut strategy for this measure using ideas from \cite{orecchia2011fast}.\footnote{The details of such a cut player do not appear in \cite{pmlr-v162-orecchia22a} or \cite{orecchia2011fast}.}

\smallskip
\looseness = -1
The structure of this paper is as follows.
Section~\ref{section:preliminaries} contains basic definitions.
In section~\ref{section:main-result} we present a version of the cut-matching game, together with a cut player and a matching player which are suitable for our expander decomposition algorithm with respect to a vertex measure $\muG$. Finally, Section~\ref{section:proof-main-theorem} contains the balanced cut algorithm, from which we obtain an expander decomposition in a standard  way which is described in Section~\ref{section:expander_decomposition}. Appendix~\ref{appendix-proofs} contains proofs that were deferred from the main text. Appendices~\ref{appendix:matrix_inequalities} and~\ref{appendix:projection} contain algebraic tools and probabilistic lemmas, respectively.

	\section{Preliminaries}
	\label{section:preliminaries}
	
    Throughout the paper we work with the undirected graph $G = (V,E)$. Let $n\defeq |V|$ and $m \defeq |E|$.
	We denote the transpose of a vector or a matrix $x$ by $x'$. That is if $v$ is a column vector then $v'$ is the corresponding row vector. 
	For a vector $v\in \RR_{\ge 0}^k$, define
	$\sqrt{v}$ to be vector whose coordinates are the square roots of those of $v$. Given $A\in \RR^{k\times l}$, we denote by $\A{}{i,j}$ the element at the $i$'th row and $j$'th column of $\A$. We denote by $\A{}{i,}, \A{}{,i}$ the $i$'th row and column of $A$, respectively. We think of both $\A{}{i,}$ and $\A{}{,i}$ as column vectors. 
	We use the abbreviation $\A{}{i}\defeq \A{}{i,}$ only with respect to the rows of $A$. Given a vector $v\in \RR^k$, we denote its $i$'th element by $\vv{i}$. For disjoint $A,B \subseteq V$, we denote by $E_G(A,B)$ the set of edges connecting $A$ and $B$. We denote by $|E_G(A,B)|$ the number of edges in $E_G(A,B)$, or the sum of their weights if the graph is weighted. We sometimes omit the subscript when the graph is clear from the context. For a weighted graph $H= (V,E,w)$, we define $\deg_H(v) = w\left(E_H(\{v\},V\setminus\{v\})\right)$.
    If $A = V \setminus B $, then we call $(A,B)$ a \emph{cut}. Throughout the paper we work with an arbitrary vertex measure $\muG: V \to \RR_{\ge 0}$.\footnote{We identify the set $V$ with $[n]$, and therefore we sometimes write $\muG\in \RR_{\ge 0}^n$.}
    For $S\subseteq V$ we let $\muG(S) = \sum_{v\in S}\muG(v)$.
    We define the set of \emph{terminals} to be $\ter = \{v \in V \mid \muG{}{v} > 0\}$. We assume the following on $\muG$. This assumption is mainly used in the proof of Lemma \ref{lemma:potential_step_2}. 
 \begin{assumption}\label{assumption:mu}
     For every $v\in \ter$ we have $\frac{1}{\poly(n)} \le \muG(v) \le \poly(n)$.
 \end{assumption}

 We show in Appendix~\ref{app:remove_assumption} how to remove Assumption~\ref{assumption:mu} and obtain refined guarantees that depends on the support size and range of the vertex measure $\muG$.

\begin{definition}[$\mu$-Matching]
    \label{def:d-matching}
    Given a vertex measure $\muG \in \RR_{\ge 0}^n$ and a collection of weighted pairs $M = \left\{ (u_i,v_i,w_i) \right\}_{i=1}^{k}$. We say that $M$ is a $\muG$-\emph{matching} if the graph defined by $M$ satisfies $\deg_{M}(v) = \dG{}{v}$, for every $v$. 
    \end{definition}
\begin{definition} [$\muGG$-stochastic]
    \label{def:d_G_stochastic}
        Let $\muGG \in \RR_{\ge 0}^n$ be a vertex measure. A matrix $\F\in \RR^{n\times n}$ is called \emph{$\muGG$-stochastic} if the following two conditions hold: (1) $F\cdot\1_n = \dGG $ and (2) $   \1_n'\cdot F = \dGG' $.
    \end{definition}
\begin{definition}[Laplacian, Normalized Laplacian]
	\label{def:laplacian_normalized_laplacian}
	    Let $A\in \RR^{n \times n}$ be a symmetric matrix and let $d = A\cdot \1_n,\; D= \diag(d)$. The Laplacian of $A$ is defined as $\mathcal{L}(A) = D - A$. The normalized-Laplacian of $A$ is defined as $\mathcal{N}(A) = D^{-\frac{1}{2}} \mathcal{L}(A) D^{-\frac{1}{2}}= I - D^{-\frac{1}{2}} A D^{-\frac{1}{2}}$. The (normalized) Laplacian of an undirected graph is defined analogously using its adjacency matrix. 
	\end{definition}
\begin{definition}[$\muG$-expansion]
	Let $G= (V,E)$ and $S \subset V$, $S\neq\emptyset$. Let $\muG \in \RR_{\ge 0}^n$. The \emph{$\muG$-expansion of the cut $(S, \comp{S})$} satisfying $\min(\muG{}{S}, \muG{}{\comp{S}}) >0$,  denoted by  $\Phi^{\muG}_{G}(S,\comp{S})$, is 
	\[
	\Phi^{\muG}_{G}(S,\comp{S}) = \frac{|E(S,\comp{S})|}{\min(\muG{}{S}, \muG{}{\comp{S}})}.
	\]
	The \emph{$\muG$-expansion of $G$} is defined to be 
    \[ \Phi^{\muG}(G) = \min_{\substack{S\subseteq V \\ \min(\muG{}{S}, \muG{}{\comp{S}})>0}}\Phi_{G}(S,\comp{S}) \].
	\end{definition}

\begin{definition}[$\muG$-Expander, $\muG$-Near-Expander]
	\label{def:expander-near-expander}
	Let $G = (V,E)$. We say that $G$ 
	is a $(\phi,\muG)$-expander if $\Phi^{\muG}(G) \ge \phi$. Let $A\subseteq V$. We say that $A$ is a \emph{near} $(\phi, \muG)$-\emph{expander}
	in $G$ if 
	\[
	\min_{\substack{S \subseteq A \\ 0<\muG{}{S}\le\muG{}{A}/2}}\frac{|E(S,V \setminus  S)|}{\min(\muG{}{S}, \muG{}{A\setminus S})}\ge \phi.
	\]
	\end{definition}

 \looseness=-1
    That is, a near expander is allowed to use cut edges that go outside of $A$. Note that a subset of a near expander is also a near expander. This notion is useful in several contexts. Specifically, using the ``trimming" algorithm of~\cite{saranurak2019expander}, under certain conditions on $A$, one can extract an expander $B\subseteq A$ by ``trimming" a part of $A$ (see Theorem~\ref{theorem:trimming}).

\begin{definition} [Embedding]
    \label{def:embedding}
        Let $G=(V,E)$ be an undirected  graph. Let $F\in \RR^{V\times V}_{\ge 0}$ be a matrix (not necessarily symmetric). We say that $F$ is \textit{embeddable} in $G$ \textit{with congestion} $c$, if there exists a multi-commodity flow $f$ in $G$, with $|V|$ commodities, one for each vertex (the vertex $v$ is the source of its commodity), such that, simultaneously for each $(u,v)\in V\times V$, $f$ routes $F(u,v)$ units of $u$'s commodity from $u$ to $v$, and the total flow on each edge is at most $c$ times the capacity of the edge. 
        
        If $F$ is the weighted adjacency matrix of a graph $H$ on the same vertex set $V$, we say that $H$ is \textit{embeddable} in $G$ \textit{with congestion} $c$ if $F$ is embeddable in $G$ with congestion $c$.\footnote{This definition requires to route $F(u,v)=F(v,u)$ units both from $u$ to $v$ and from $v$ to $u$ if $F$ is symmetric.} 
    \end{definition}
\begin{lemma}
    \label{lemma:near_expansion_and_embedding}
    \looseness = -1
        Let $G,H$ be two graphs on the same vertex set $V$. Let $A\subseteq V$. Assume that $H$ is embeddable in $G$ with congestion $c$, and that $A$ is a near $(\phi,\mu)$-expander in $H$. Then, $A$ is a near $(\frac{2\phi}{c}, \mu)$-expander in $G$.
    \end{lemma}
    \begin{proof}
        Let $S\subseteq A, \comp{S} = V\setminus S$ be a cut, and assume $0<\muG{}{S}\le\muG{}{A\setminus S}$. 
        In the embedding of $H$ in $G$, each edge $(u,v) \in E(H)$ corresponds to $w_H(u,v)$ units of flow, routed in $G$ from $u$ to $v$. Since each edge of $G$ carries at most $c$ units of flow (times its capacity, for weighted graphs)
        we get that $|E_G(S,\comp{S})|\ge \frac{2}{c}|E_H(S,\comp{S})|$.\footnote{Recall that by definition for every $\{u,v\}\in E(H)$ we send flow in $G$ in both directions.}
        Since $A$ is a near $(\phi, \muG)$-expander in $H$, we get that
        
        \begin{align*}
            \Phi^{\muG}_G(S,\comp{S}) = \frac{|E_G(S,\comp{S})|}{\muG{}{S}}\ge\frac{2}{c}\frac{|E_H(S,\comp{S})|}{\muG{}{S}} \ge \frac{2\phi}{c} .
        \end{align*}
    \end{proof}

\begin{corollary}
    \label{cor:expansion_and_embedding}
        Let $G,H$ be two graphs on the same vertex set $V$. Assume that $H$ is embeddable in $G$ with congestion $c$, and that $H$ is a $(\phi,\mu)$-expander. Then, $G$ is a $(\frac{2\phi}{c},\mu)$-expander.
    \end{corollary}
    \begin{proof}
        This follows from Lemma~\ref{lemma:near_expansion_and_embedding} by choosing $A = V$.
    \end{proof}

\section{Balanced Sparse-Cut via spectral Cut-Matching}
    \label{section:main-result}
In this section we state our main theorem and present our algorithm, whose analysis appears in the next section. The result of~\cite{agassy2022expander} for finding a balance sparse cut for the case of conductance is a special case of our theorem in which the vertex measure corresponds to degrees (\ie{} $\muG(v) = \deg_{G}(v)$ for every $v\in V$).

\begin{theorem}
    \label{theorem:main-result}
    \label{theorem:cut_matching}
    Given a graph $G=(V,E)$ with $n$ vertices and $m$ edges, a vertex measure $\muG\in \RR_{\ge 0}^n$ that satisfies Assumption~\ref{assumption:mu} and a parameter $\phi>0$,
        there exists a randomized algorithm which takes $\tilde{O}\left(m\right)$
        time\footnote{The $\tilde{O}$ hides $\log n$ factors .}
        and must end with one of the following three cases:
        \begin{enumerate}
            \item We certify that $G$ has $\muG$-expansion $\Phi^{\muG}(G)=\Omega(\phi)$.
            \item We find a cut $(R,A)$ in $G$ of $\muG$-expansion $\Phi^{\muG}_G(R,A)=O(\phi\log n)$, and $\muG{}{R}, \muG{}{A}$ are both $\Omega\!\left(\frac{\muG{}{V}}{\log n}\right)$, i.e, we find a relatively balanced low conductance cut.
            \item We find a cut $(R,A)$ with $\Phi^{\muG}_G(R,A)\le c_0 \phi \log n$ for some constant $c_0 > 0$, $0<\muG{}{R}\le \frac{\muG{}{V}}{10c_0\log n}$, and $A$ is a near $(\Omega(\phi),\muG)$-expander in $G$.
        \end{enumerate}
        The correctness of the algorithm holds with high probability.\footnote{By ``with high probability" we mean a probability of at least $1-\frac{1}{n^\ell}$, where $\ell$ is an arbitrary constant that affects our running time linearly.}.
    \end{theorem}

This theorem is a straightforward generalization of~\cite[Theorem 5.1]{agassy2022expander}. The main algorithmic difference is the usage of fractional matchings during the cut matching game and zeroing out rows and columns in the random walk matrix that correspond to vertices $v\notin \ter$. (Recall that
$\ter = \{v \in V \mid \muG{}{v} > 0\}$.)

To get Theorem~\ref{theorem:cut_matching}
we use the cut matching game, specified in the next subsection. Theorem~\ref{theorem:cut_matching} implies the following theorem.
\begin{theorem}
    \label{theorem:expander_decomposition}
        Given a graph $G = (V, E)$ with $n$ vertices and $m$ edges, a vertex measure $\muG\in \RR^n_{\ge 0}$ that satisfies Assumption \ref{assumption:mu}, and a parameter $\phi > 0$, there is a randomized algorithm that with high probability finds a partition of $V$ into clusters $V_1,..., V_k$ such that $\forall i: \Phi^{\muG}_{G[V_i]} = \Omega(\phi)$ and $\sum_{i}{|E(V_i, \comp{V_i})|} = O(\phi \muG{}{V} \log^2 n)$.
        The running time of the algorithm is $\tilde{O}(m)$.
    \end{theorem}

    The derivation of Theorem~\ref{theorem:expander_decomposition}
    using Theorem~\ref{theorem:cut_matching} uses the same methods as in \cite{saranurak2019expander, LNPSsoda13, agassy2022expander}. For completeness, we give this derivation in Section~\ref{section:expander_decomposition}. 

    The rest of this section is organized as follows. In Section~\ref{subsection:cut_matching} we introduce a generalized version of the cut matching game of~\cite{agassy2022expander}. In Sections~\ref{section:cut-player} and~\ref{section:matching-player} we present a cut player and a matching player, respectively. Finally, in Section~\ref{section:main-result-algorithm} we present the algorithm for Theorem~\ref{theorem:expander_decomposition} using the cut and matching players.

\subsection{Cut matching for \texorpdfstring{$\muG$}{μ}-expansion}\label{subsection:cut_matching}
    The following is a generalized version of the cut matching games in~\cite{saranurak2019expander,agassy2022expander}, suited for $\muG$.

    \begin{center}
    \fbox{\parbox{0.95\textwidth}{
        The Cut-Matching game for $\muG$-expansion, with parameters $T$ 
        and a vector $\dG$:
        \begin{itemize}
            \item The game is played on a series of graphs $G_i$. Initially, $G_0 = \emptyset, A_0=V$.
            \item At iteration $t$, the cut player produces two sets, $L_t,R_t\subseteq A_t$, and weights $m_t: L_t \to\RR_{\ge 0}, \bar{m}_t: R_t\to\RR_{\ge 0}$.\footnotemark 
            ~The weights satisfy $m_t(L_t) \le \frac{\muG{}{A_t}}{8}, \bar{m}_t(R_t)\ge \frac{\muG{}{A_t}}{2}$, and the total weight $m_t(v)+\bar{m}_t(v)$ of each $v\in L_t\cup R_t$ is at most $\dG{}{v}$.
            \item The matching player responds with a set $S_t\subseteq A_t$ and a $\mu$-matching $\M{t}$ that only matches vertices in $L_t\setminus S_t$ to vertices in $R_t \setminus S_t$.
            \item We set $G_{t+1} = G_t \cup \M{t}, A_{t+1} = A_t \setminus S_t$. The union $G_t \cup \M{t}$ refers to the graph whose weighted adjacency matrix is the sum of the weighted adjacency matrices of $G_t$ and $\M{t}$. 
            \item The game ends at iteration $T$, and the \emph{quality} of the game is $r$, the near $\muG$-expansion of $A_t$ in $G_t$.
        \end{itemize}
    }}
    \footnotetext{We denote $m_t(v) = 0$ for each $v\notin L_t$ and $\bar{m}_t(v) = 0$ for each $v\notin R_t$.}
    \end{center}  
    \medskip

    Our cut matching game iteratively shrinks the domain $\A{t}\subseteq V$. So we start by introducing the following matrices that would be useful to work with a shrinking domain.
 
    \begin{definition}[$\I{t},\dG{t},\U{t},\Pmat{t}$]
        We define the following variables
        \begin{enumerate}
            \item $\I{t} = \diag(\1_{A_t \cap \ter})\in\RR^{n\times n}$. Similarly, $\1_t = I_t \cdot \1$.
            \item $\dG{t} = \I{t} \cdot \dG \in \RR^n$, i.e the projection of $\dG$ onto $A_t$.
            \item $\U{t} = \I{t} \cdot U = \diag(\dG{t}) \in \RR^{n\times n}$, where $U = \diag(\muG)$.  
            \item $\Pmat{t} = \I{t} - \frac{1}{\muG{}{A_t}}\sqrt{\dG{t}}\sqrt{\dG{t}'} \in \RR^{n\times n}$.
            This matrix projects on the orthogonal complement of $\sqrt{\dG{t}}$.
        \end{enumerate}
    \end{definition}

    \begin{remark}
        Throughout the paper we abuse the inverse notation and write $\U^{-1}$ (or $\Uminushalf$) to the denote the pseudo inverse of $U$ (i.e. inverse each entry on the diagonal if it is not zero).
    \end{remark}

\subsection{Cut player}
    \label{section:cut-player}
    Let $\delta = \Theta(\log{n})$ 
    be a power of $2$ which we formally set later in the proof of Lemma~\ref{lemma:F_expander}.\footnote{We need $\delta$ to be a power of $2$ to apply Theorem \ref{theorem:symmetric_rearrangement} to a $\delta$ matrix power later on.
    } 
    The cut player implicitly  maintains a $\dGG$-stochastic matrix $\F{t}\in \RR^{n\times n}$, and the graph $G_t$ which is the union of the matchings that it obtained so far from the matching player ($t$ is the index of the round). The matrix $F_t$ and the graph $G_t$ have two crucial properties.
    First, we can embed (recall Definition~\ref{def:embedding}) $\F{t}$ in $G_t$ with  $O(\frac{1}{\delta})$ congestion (Lemma~\ref{lemma:F_t_embeddable_in_G_t}) and embed $G_t$ in $G$ with $O(\frac{t}{\phi \log{n}})$ congestion (Lemma~\ref{lemma:G_t_embeddable_in_G}). Second, after
    $T=\Theta(\log^2 n)$ rounds, $A_t$ will be a  near $(\Omega(1),\muG)$-expander in $\F{T}$ with high probability (Lemma~\ref{lemma:F_expander}). By Lemma~\ref{lemma:near_expansion_and_embedding}, this implies that, with high probability, $A_t$ will be a  near $(\Omega(\phi),\muG)$-expander in $G$.

    We define the matrix  $\W{t} = (\Pmat{t}\normalized{\F{t}}\Pmat{t})^{\delta}$ (we recall again that $\delta = \Theta(\log{n})$ as specified above). This definition makes us ``focus'' only on the remaining vertices $A_t \cap \ter$, as any row/column of $\W{t}$ corresponding to a vertex $v\in V \setminus (A_t\cap \ter)$ is zero. 
    The matrix $W_t$ is used in this section to define the projections that the cut player needs to update $F_t$. It is also used in Section~\ref{section:A_expander_in_F} to define the potential that measures how far is the remaining part of the graph $A_T$ from a near $\muG$-expander. In particular, we show in Lemma~\ref{lemma:F_expander} and Corollary~\ref{cor:G_expander} that if $W_T^2$ has small eigenvalues (which will be the case when the potential is small) then $A_T$ is near $(\Omega(\delta), \muG)$-expander in $G_T$.
    
    The cut player updates $\F{t}$ as follows. Let $r\in \RR^n$ be a random unit vector. 
    Consider the projections $u_i = \frac{1}{\sqrt{\dG{}{i}}}\langle \W{t}{i}, r \rangle$, for $i\in A_t \cap \ter$ and $u_i = 0$ otherwise. Note that since $\Pmat{t} \sqrt{\dG{t}} = 0$, and $\W{t}$ is symmetric we have that:
    \begin{align*}
        \sum_{i\in A_t}{\mu(i) u_i} = \sum_{i\in A_t}{\sqrt{\mu(i)}\left\langle\W{t}{i}, r\right\rangle} = \left\langle \sum_{i\in A_t}{\sqrt{\mu(i)}\W{t}{i}}, r \right\rangle
        = \left\langle \W{t}\sqrt{\dG} , r \right\rangle = 0
    \end{align*}

    We use the following lemma to partition (some of) the vertices of $A_t$ into two weighted sets $A_t^l$ and $A_t^r$.\footnote{Note that this does not produce a bisection of $A_t$.} The lemma is a continuous modification of~\cite[Lemma 3.3]{racke2014computing}. For completeness, we include the proof in Appendix~\ref{appendix-proofs}.

    \begin{lemma} [\texorpdfstring{\cite[Lemma 3.3]{racke2014computing}}{[RST14, Lemma 3.3]}]
	\label{lemma:RST}
	    Given $u_i\in \RR$ for all $i\in A_t$, such that $\sum_{i\in A_t}{\dG{}{i} u_i} = 0$, 
	    we can find in time $O(|A_t|\cdot \log (|A_t|))$ disjoint sets of source nodes $A^l_t\subseteq A_t$ and target nodes $A^r_t\subseteq A_t$, weights $\mu(i)/4\le m_t(i) \le \mu(i)$ for all $i\in A^l_t$, $\mu(i)/4 \le \bar{m}_t(i) \le \mu(i)$ for all $i \in A^r_t$, and a separation value $\eta \in \RR$ such that: 
	    \begin{enumerate}
	        \item $\eta$ separates the sets $A^l_t, A^r_t$, \ie,\ either $\max_{i\in A^l_t}{u_i}\le \eta \le \min_{j\in A^r_t}{u_j}$, or $\min_{i\in A^l_t}{u_i}\ge \eta \ge \max_{j\in A^r_t}{u_j}$,
            \item $\forall i\in V: m_t(i) + \bar{m}_t(i) \le \muG{}{i}$,
            \item $\bar{m}_t(A^r_t) \ge \frac{\muG{}{A_t}}{2}$, $m_t(A^l_t) \le \frac{\muG{}{A_t}}{8}$,
	        \item $\forall i\in A^l_t: (u_i-\eta)^2\ge \frac{1}{9}u_i^2$, 
	        \item $\sum_{i\in A^l_t}{m_t(i) u_i^2}\ge \frac{1}{80}\sum_{i\in A_t}{\dG{}{i} u_i^2}$. 
	    \end{enumerate}
	\end{lemma}

    The cut player sends $A_l,A_r$ and $A_t$ to the matching player. 
    
    In turn, the matching player (see Subsection~\ref{section:matching-player}) returns a cut $(S_t,A_t\setminus S_t$) and a fractional matching $M_t$ of total weight $m_t(A^l_t \setminus S_t)$ between $A^l_t \setminus S_t$ and $A^r_t \setminus S_t$. 
    Define $\N{t} = \frac{\delta - 1}{\delta}U + \frac{1}{\delta}\M{t}$, the lazy random walk that takes a step according to $M_t$ with probability $\frac{1}{\delta}$, and otherwise takes a step according to $U$.
    The cut player then updates $\F{{t}}$ as follows: $\F{{t+1}} = \N{t} \cdot U^{-1} \F{t} U^{-1} \N{t}$. Intuitively, this applies a weighted average of the rows and columns of $\F{t}$ according to $\M{t}$, where each row/column ``keeps" $\frac{\delta-1}{\delta}$ fraction of its value.
    For the analysis, we update $G_{t+1}$ as $G_{t+1} = G_t\cup \M{t}$.\footnote{We allow parallel edges.}

    \subsection{Matching player}\label{section:matching-player}

    The matching player receives $A^l_t$ and $A^r_t$ and the current $A_t$. 
    For a vertex $v\in V$, denote by $m_t(v)$ the weight of $v$ in $A_t^l$, and by $\bar{m}_t(v)$ the weight of $v$ in $A_t^r$ ($m_t(v) + \bar{m}_t(v) \le \muG{}{v}$).
    The matching player solves the flow problem on $G[A_t]$, specified by the following lemma, that closely resembles~\cite[Lemma 8.1]{LNPSsoda13} when setting $V = \A_t$.
    For completeness, we include the proof in Appendix~\ref{appendix:omitted-proofs-2}.
    
    \begin{lemma}[\texorpdfstring{\cite[Lemma 8.1]{LNPSsoda13}}{[LNPS23, Lemma 8.1]}]
	\label{lemma:fair-cuts}
	    Let $G=(V,E)$ be a graph with $n$ vertices, let $A^l, A^r \subseteq V$ be disjoint, and let $\phi>0$ be a parameter. For a vertex $v\in V$, denote by $m(v)$ the weight of $v$ in $A^l$, and by $\bar{m}(v)$ the weight of $v$ in $A^r$. Assume that $m(v) + \bar{m}(v)\le \dG{}{v}$, that each positive weight $m(v)$ or $\bar{m}(v)$ is at least $\muG{}{v}/4$, and that $\bar{m}(A^r) \ge \frac{1}{2}\muG{}{V}, m(A^l) \le \frac{1}{8}\muG{}{V}$. We define the flow problem $\Pi(G)$, as the problem in which a source $s$ is connected to each vertex $v\in A^l$ with an edge of capacity $m(v)$ and each vertex $v\in A^r$ is connected to a sink $t$ with an edge of capacity $\bar{m}(v)$. Every edge of $G$ has the same capacity $c=c_1/(\phi\log n)$, where $c_1>0$ is a fixed constant chosen in the proof of Corollary~\ref{cor:G_expander}. 
        A feasible flow for $\Pi(G)$ is a maximum flow that saturates all the edges outgoing from $s$.
	    Then, there exists a randomized algorithm, which in time $\tilde{O}(|E|)$, finds either
	    \begin{enumerate}
	        \item A feasible flow $f$ for $\Pi(G)$; or
	        \item A cut $S$ where $\Phi_{G}^{\mu}(S, V\setminus S)\le\frac{7}{c}=O(\phi\log n)$, $\muG{}{V\setminus S} \ge \frac{1}{3}\muG{}{V}$ and a feasible flow for the problem $\Pi(G[V\setminus S])$, where we only consider the sub-graph $G[V\setminus S \cup \{s, t\}]$ (that is, vertices $v\in A^l\setminus S$ are sources of $m(v)$ units, and vertices $v\in A^r\setminus S$ are sinks of $\bar{m}(v)$ units). 
	    \end{enumerate}
        The correctness of the algorithm holds with high probability.
	\end{lemma}
   
	\begin{remark}
	    It is possible that $A^l\subseteq S$, in which case the feasible flow for $\Pi(G[V\setminus S])$ is trivial (the total source mass is $0$). 
	\end{remark}
    Let $S_t$ be the cut returned by the lemma. If the lemma terminates with the first case, we set $S_t = \emptyset$. 
    Using dynamic trees~\cite{ST83},
    we can decompose the returned flow into a set of (at most $m$) weighted paths, each carrying some flow from a vertex $u\in A_t^l\setminus S_t$ to a vertex $v\in A^r_t\setminus S_t$. 
    Each $u\in A^l_t\setminus S_t$ is the starting point of exactly $m_t(u)\le\dG{}{u}$ units of flow, and each $v\in A^r_t\setminus S_t$ is the target of \textit{at most} $\bar{m}_t(v)\le\dG{}{v}$ units of flow. Define the weighted matching
    $\tilde{\M{t}}$ as $\tilde{\M{t}} = ((u_i, v_i, w_i))_{i=1}^{|A^l_t\setminus S_t|}$, where $u_i$ and $v_i$ are the endpoints of path $i$ which carries $w_i$ flow. We can view $\tilde{\M{t}}$ as a symmetric $n\times n$ matrix, such that $\MatrixSimple{\tilde{\M{t}}}{u, v} = \MatrixSimple{\tilde{\M{t}}}{v, u} =w_{uv}$, where $w_{uv}$ is the total flow sent on paths between $u$ and $v$. We turn $\tilde{\M{t}}$ into a $\dG$-stochastic matrix by increasing its diagonal entries by $\dG - \tilde{\M{t}}\1_n$. Formally, we set $\M{t} \defeq \tilde{\M{t}} + \diag(\dG - \tilde{\M{t}}\1_n)$. Notice that $\dG - \tilde{\M{t}}\1_n$ has only non-negative entries, so $\M{t}$ also has non-negative entries.
   The matching player sends back $\M{t}, S_t$ as the response to the subsets $A^l_t$ and $A^r_t$ given by the cut player. 

\subsection{The Algorithm}\label{section:main-result-algorithm}
    Let $T=\Theta(\log^2 n)$ and $c=c_1/(\phi\log n)$. 
    The algorithm for Theorem~\ref{theorem:main-result} is presented in Algorithm~\ref{algo:cut_matching}.  
    The algorithm runs for at most $T$ rounds and stops when $\muG{}{R_t}>\frac{\muG{}{V}\cdot c\cdot \phi}{70}=\Omega(\frac{\muG{}{V}}{\log n})$. 
    In each round $t$, we implicitly update $\F{t}$ (see Section~\ref{section:cut-player}). In order to keep the running time near linear, we compute the flow of Lemma~\ref{lemma:fair-cuts} using the fair cut result of~\cite{LNPSsoda13}. The fair-cut routine may also return a cut $S_t \subseteq A_t$ which satisfies $\Phi_{G[A]}(S_t, A_t\setminus S_t) \le \frac{7}{c}$ (with high probability), 
    in which case we ``move'' $S_t$ from $A_{t+1}$ to $R_{t+1}$. After $T$ rounds, $\F{T}$ certifies that the remaining part of $A_T$ is a near $(\Omega(\phi),\mu)$-expander with high probability.

    \begin{algorithm}[hbt!]
        \caption{Cut-Matching}
        \label{algo:cut_matching}
        \begin{algorithmic}[1]
            \Function{Cut-Matching}{$G, \phi, \muG$}
                \State $T\gets\Theta(\log^2 n)$, $c\gets c_1/(\phi\log n)$.
                \State $t\gets 0$.
                \State $A_0\gets V$,\; $R_0 \gets \emptyset$. \;
                \State Set $\F{0}\gets U \defeq \diag(\dG)$.
                \While{$\muG{}{R_t}\le \frac{\muG{}{V}\cdot c\cdot \phi}{70}$ and $t < T$} 
                    \State Update $\F{t}$ to $\F{t+1}$. \Comment{Sections~\ref{section:cut-player} and~\ref{section:matching-player} describe this update.}
                    \State This update returns $S_t\subseteq A_t$ where $\Phi^{\muG}_{G[A_t]}(S_t, A_t\setminus S_t)\le \frac{7}{c}$, or $S_t=\emptyset$.
                    \State $A_{t+1} \gets A_t - S_t$, $R_{t+1} \gets R_t \cup S_t$.
                    \State $t\gets t+1$.
                \EndWhile
                \If {$t = T$}
                    \If {$R_t = \emptyset$}
                        \State Certify that $\Phi^{\muG}(G)=\Omega(\phi)$. \Comment{Case (1) of Theorem~\ref{theorem:main-result}.}
                    \ElsIf {$\muG{}{R_t} > \frac{\muG{}{V}\cdot c\cdot \phi}{70}$}
                        ~\Return {$(A_T, R_T)$}. \Comment{Case (2) of Theorem~\ref{theorem:main-result}.}
                    \Else 
                        ~\Return {$(A_T, R_T)$}. \Comment{Case (3) of Theorem~\ref{theorem:main-result}.}
                    \EndIf                
                \Else
                    \State \Return {$(A_t, R_t)$}. \Comment{Case (2) of Theorem~\ref{theorem:main-result}.}
                \EndIf
            \EndFunction
        \end{algorithmic}
    \end{algorithm}

\section{Proof of Theorem \texorpdfstring{\ref{theorem:cut_matching}}{3.1}}\label{section:proof-main-theorem}

    In this section we prove that Algorithm~\ref{algo:cut_matching} satisfies the conditions of Theorem~\ref{theorem:cut_matching}.
    We start by giving a high-level overview of the analysis of our algorithm.
    The potential we use here to measure the progress of the cut player is $\psi(t) = \tr[\W{t}^2] = \sum_{i\in A_t}{\norm{\W{t}{i}}_2^2}$, where $\W{t}$ was defined as $\W{t} = (\Pmat{t} \Dminushalf \F{t} \Dminushalf \Pmat{t})^\delta$. Recall $\delta=O(\log n )$ is a power of $2$, which is specified in the proof of Lemma~\ref{lemma:F_expander}.

    Our key technical Lemmas~\ref{lemma:potential_step_1} and~\ref{lemma:potential_step_2} carefully establish the decrease in potential achieved by the cut player, taking into account the matchings and the deleted cuts. We conclude in Corollary~\ref{cor:total_potential} that after $O(\log^2 n)$ iterations the potential falls below $1/n$ with high probability. We then prove, using our variation on Cheeger's inequality (see Lemma~\ref{lemma:F_expander}), that having a small potential means that $A_t$ is a near $(\Omega(1),\muG)$-expander in $F_t$. 

    Using the fact that $\F{t}$ is embeddable in $G_t$ (Lemma~\ref{lemma:F_t_embeddable_in_G_t}) and that $G_t$ is embeddable in $G$ (Lemma~\ref{lemma:G_t_embeddable_in_G}) we can conclude that if we reach round $T$, then with high probability, $A_T$ is a near $(\Omega(\phi), \mu)$-expander in $G$ (Corollary~\ref{cor:G_expander}). 

    This section is organized as follows. Subsection~\ref{section:main-result_F_embeddable_G} shows that $\F{t}$ is embeddable in $G_t$ with congestion $\frac{4}{\delta}$ and that $G_t$ is embeddable in $G$ with congestion $2c\cdot t$. Subsection~\ref{section:A_expander_in_F} shows that if we reach round $T$, then with high probability, $A_T$ is a near $(\Omega(\phi),\mu)$-expander in $G$. Finally, in Subsection~\ref{section:main-result_theorem_proof} we prove Theorem~\ref{theorem:main-result}.
 
    \subsection{\texorpdfstring{$\F{t}$}{Ft} is embeddable in \texorpdfstring{$G$}{G}}
    \label{section:main-result_F_embeddable_G}
    
    To begin the analysis of the algorithm, we first define a blocked matrix. This notion will be useful when our matrices ``operate'' only on vertices of $A_t$.
    \begin{definition}
    \label{def:d_t_block_stochastic}
        Let $A \subseteq V$. A matrix $B\in \RR^{n\times n}$ is $A$-\emph{blocked} if $\MatrixSimple{B}{i, j} = 0$ for all $i\neq j$ such that $(i,j)\notin A \times A$. Diagonal entries may be non-zero.
    \end{definition}
    \begin{lemma}  
    \label{lemma:basic_properties}
        The following holds for all $t$:
        \begin{enumerate}
            \item $\M{t}, \N{t}, \F{t}$ and $\W{t}$ are symmetric.
            \item $\M{t}$ and $\N{t}$ are $\A{t} \cap \ter$-blocked. $\F{t}$ is $\ter$-blocked.
            \item $\M{t}, \N{t}$ and $\F{t}$ are $\dG$-stochastic. 
        \end{enumerate}
    \end{lemma}
    \begin{proof}
        \begin{enumerate}
            \item This is clear from the definitions.
            \item Since $\tilde{\M{t}}$ is a matching of vertices in $A_t^l\setminus S_t \subseteq A_t$ to vertices in $A_t^r\setminus S_t\subseteq A_t$, then $\MatrixSimple{\tilde{\M{t}}}{i, j} = 0$ for all $i\neq j$ such that $(i,j) \notin \A{t} \times \A{t}$. After the addition of $\diag(\muG - \tilde{\M{t}}\1_n)$, we still have $\M{t}{i, j} = 0$ for all $i\neq j$ such that $(i,j) \notin \A{t} \times \A{t}$, so $\M{t}$ is $\A{t}$-blocked. Now, for $\N{t}$, note that for all $i\neq j$ such that $(i,j) \notin \A{t} \times \A{t}$, $\N{t}{i, j} = \frac{\delta - 1}{\delta}\MatrixSimple{U}{i, j} + \frac{1}{\delta}\M{t}{i, j} = \frac{\delta - 1}{\delta}\cdot 0 + \frac{1}{\delta}\cdot 0 = 0$.

            Finally, $\F{t}$ is $\ter$-blocked because it a product of $\ter$-blocked matrices. 
            \item 
            For $\M{t}$ this is true because we obtain it by explicitly making $\tilde{\M{t}}$  $\dG$-stochastic. For $\N{t}$, note that
            \begin{align*}
                \N{t} \1_n = \left(\frac{\delta - 1}{\delta}\U + \frac{1}{\delta}\M{t}\right) \1_n = \frac{\delta - 1}{\delta}\dG + \frac{1}{\delta}\dG = \dG
            \end{align*}
            and $\1_n' \N{t} = \dG'$ follows from symmetry. For $\F{t}$, we use induction on $t$. $\F{0} = \U = \diag(\dG)$ is clearly $\dG$-stochastic. After step $t$, 
            \begin{align*}
                \F{{t+1}}\1_n = \N{t} U^{-1} \F{t} U^{-1} \N{t} \cdot\1_n = \N{t} U^{-1} \F{t} U^{-1} \cdot \dG  = \N{t} U^{-1} \F{t} \1_{\ter} \numeq{4}  \N{t} U^{-1} \cdot \dG = \N{t}\1_{\ter} \numeq{6} \dG,
            \end{align*}
            where Equalities $(4)$ and $(6)$ follow from part $2$ of Lemma~\ref{lemma:basic_properties}. We get $\1_n' \F{{t+1}} = \dG'$ from symmetry.  
        \end{enumerate}
    \end{proof}

    The following lemmas show that $F_t$ is embeddable in $G_t$ and that $G_t$ is embeddable in $G$. The proof of Lemma~\ref{lemma:F_t_embeddable_in_G_t_step1} is deferred to Appendix~\ref{appendix-proofs}.

    \begin{lemma}
    \label{lemma:F_t_embeddable_in_G_t_step1}
    Let $F\in \RR^{n\times n}$ be a $\dG$-stochastic matrix, and let $H$ be a graph. Fix any $t$. Assume that $F$ is embeddable in $H$ such that the congestion on each edge $e\in H$ is $c(e)$, then
    \begin{enumerate}
        \item $\N{t}\cdot U^{-1} F$ is embeddable in $H\cup\M{t}$ such that the congestion on edges of $\M{t}$ is $\frac{2}{\delta}$ and the congestion on each edge $e\in H$ is still $c(e)$.
        \item $F U^{-1} \cdot \N{t}$ is embeddable in $H\cup\M{t}$ such that the congestion on edges of $\M{t}$ is $\frac{2}{\delta}$ and the congestion on each edge $e\in H$ is still $c(e)$.
        \item $\N{t} \cdot U^{-1} \F{t} U^{-1} \N{t}$ is embeddable in $H \cup \M{t}$ such that the congestion on edges of $\M{t}$ is $\frac{4}{\delta}$ and the congestion on each edge $e\in H$ is still $c(e)$.
    \end{enumerate}
    \end{lemma}

    \begin{lemma}
    \label{lemma:F_t_embeddable_in_G_t}
        For all rounds $t$, $\F{t}$ is embeddable in $G_t$ with congestion $\frac{4}{\delta}$.
    \end{lemma}
    \begin{proof}
        This is a direct consequence of Lemma~\ref{lemma:F_t_embeddable_in_G_t_step1}.
    \end{proof}
    \begin{lemma}
    \label{lemma:G_t_embeddable_in_G}
        For all rounds $t$, $G_t$ is embeddable in $G$ with congestion $2ct$.
    \end{lemma}
    \begin{proof}
        For every $t$, routing the flow paths from round $t$ in both directions embeds $\M{t}$ in $G$ with congestion $2c$. Summing these routings gives a routing of $G_t = \bigcup_{i=1}^t{\M{i}}$ in $G$ with congestion $2c\cdot t$.
    \end{proof}

    \subsection{\texorpdfstring{$A_T$}{AT} is a near expander in \texorpdfstring{$\F{T}$}{FT}}
    \label{section:A_expander_in_F}
    In this section we prove that after $T=\Theta(\log^2 n)$ rounds, with high probability, $A_T$ is a near $(\Omega(1), \muG)$-expander in $\F{T}$, which implies that it is a near $(\Omega(\phi), \muG)$-expander in $G$.
    
    The section is organized as follows. Lemma~\ref{lemma:submatrix_properties} proves matrix identities and Lemma~\ref{lemma:X_as_normalized_laplacian} specifies a spectral property that our proof requires.
    We then define a potential function and lower bound the decrease in potential in Lemmas~\ref{lemma:potential_step_1}-\ref{cor:total_potential}. Finally, in Lemma~\ref{lemma:F_expander} and Corollary~\ref{cor:G_expander} we use the upper bound on the potential at round $T$, to show that with high probability $A_T$ is a near $(\Omega(1), \muG)$-expander in $\F{T}$ and a near $(\Omega(\phi), \muG)$-expander in $G$. 
    \begin{lemma}
    \label{lemma:submatrix_properties}
    The following relations hold for all $t$:
        \begin{enumerate}
            \item For any $\A{t}\cap \ter$-blocked matrix $B\in \RR^{n\times n}$ that is either $\dG{t}$-stochastic or $\dG$-stochastic,  we have $\I{t}\normalized{B} = \normalized{B}\I{t}$ and $\Pmat{t} \cdot  \Dminushalf B \Dminushalf =  \Dminushalf B \Dminushalf \cdot \Pmat{t}$.   
            \item $\I{t}\Pmat{t} = \Pmat{t}$, $\I{t}^2 = \I{t}$ and $\Pmat{t}^2 = \Pmat{t}$.
            \item $\Pmat{t} \Pmat{{t+1}} = \Pmat{{t+1}} \Pmat{t} = \Pmat{{t+1}}$.
            \item $\Pmat{t} = \Dminushalf \mathcal{L}(\frac{1}{\muG{}{A_t}} \dG{t} \dG{t}') \Dminushalf.$
            \item for any $v\in \RR^n$, it holds that 
            $v' \mathcal{L}\left(\frac{1}{\muG{}{A_t}} \dG{t} \dG{t}'\right) v = \norm{\U{t}^{\frac{1}{2}} v}_2^2 - \frac{1}{\muG{}{A_t}}  \left\langle v, \dG{t}\right\rangle^2$. 
            \item For any $B\in \RR^{n \times n}$, $\tr(\I{t} BB') = \sum_{i\in A_t \cap \ter}\norm{B(i)}_2^2$.
        \end{enumerate}
    \end{lemma}
    \begin{proof}
        \begin{enumerate}
            \item Since $B$ is $\A{t}\cap\ter$-blocked, $\I{t}\normalized{B}=\normalized{B}\I{t}$. Under either stochasticity assumption, the block property gives
            \[
                B\1_t=\dG{t},
                \qquad
                \1_t'B=\dG{t}'.
            \]
            Consequently,
            \begin{align*}
                \Dminushalf B \Dminushalf \sqrt{\dG{t}}\sqrt{\dG{t}'} = 
                \Dminushalf B \1_t \sqrt{\dG{t}'} = \Dminushalf \dG{t} \sqrt{\dG{t}'} = \sqrt{\dG{t}}\sqrt{\dG{t}'}, \\ 
                \sqrt{\dG{t}}\sqrt{\dG{t}'} \Dminushalf B \Dminushalf  = 
                \sqrt{\dG{t}} \1_t'  B  \Dminushalf = \sqrt{\dG{t}}  \dG{t}' \Dminushalf = \sqrt{\dG{t}}\sqrt{\dG{t}'} \ .
            \end{align*}
            
            Thus,
            \begin{align*}
                \Pmat{t} \cdot \Dminushalf B \Dminushalf &= (\I{t}-\frac{1}{\muG{}{A_t}}\sqrt{\dG{t}}\sqrt{\dG{t}'})\Dminushalf B \Dminushalf \\
                &= \Dminushalf B \Dminushalf \I{t}- \frac{1}{\muG{}{A_t}}\sqrt{\dG{t}}\sqrt{\dG{t}'} = \Dminushalf B \Dminushalf \cdot \Pmat{t} \ .
            \end{align*}
            
            \item The first and second equalities are clear from the definitions. For $\Pmat{t}$,
            \begin{align*}
                \Pmat{t}^2 &= \left(\I{t} - \frac{1}{\muG{}{A_t}}\sqrt{\dG{t}}\sqrt{\dG{t}'}\right)\cdot\left(\I{t} - \frac{1}{\muG{}{A_t}}\sqrt{\dG{t}}\sqrt{\dG{t}'}\right)
                \\
                &= \I{t} - \frac{2}{\muG{}{A_t}}\sqrt{\dG{t}}\sqrt{\dG{t}'} + \frac{1}{\muG{}{A_t}^2}\sqrt{\dG{t}}\sqrt{\dG{t}'}\sqrt{\dG{t}}\sqrt{\dG{t}'} = \I{t} - \frac{1}{\muG{}{A_t}}\sqrt{\dG{t}}\sqrt{\dG{t}'} = \Pmat{t}
            \end{align*}
            \item We show $\Pmat{t} \Pmat{{t+1}} = \Pmat{{t+1}}$. The other direction follows by symmetry.
            \begin{align*}
                \Pmat{t} \Pmat{{t+1}} &= \left(\I{t} - \frac{1}{\muG{}{A_t}}\sqrt{\dG{t}}\sqrt{\dG{t}'}\right)\cdot \left(\I{{t+1}} - \frac{1}{\muG{}{A_{t+1}}}\sqrt{\dG{t+1}}\sqrt{\dG{t+1}'}\right) \\&=
                \I{{t+1}} - \frac{1}{\muG{}{A_{t+1}}}\sqrt{\dG{t+1}}\sqrt{\dG{t+1}'} -
                \frac{1}{\muG{}{A_t}}\sqrt{\dG{t}}\sqrt{\dG{t}'}\I{{t+1}} +
                \frac{1}{\muG{}{A_t} \cdot \muG{}{A_{t+1}}}\sqrt{\dG{t}}\sqrt{\dG{t}'}\sqrt{\dG{t+1}}\sqrt{\dG{t+1}'} \\&=
                \I{{t+1}} - \frac{1}{\muG{}{A_{t+1}}}\sqrt{\dG{t+1}}\sqrt{\dG{t+1}'} -
                \frac{1}{\muG{}{A_t}}\sqrt{\dG{t}}\sqrt{\dG{t+1}'} +
                \frac{1}{\muG{}{A_t} }\sqrt{\dG{t}}\sqrt{\dG{t+1}'} \\&=
                \I{{t+1}} - \frac{1}{\muG{}{A_{t+1}}}\sqrt{\dG{t+1}}\sqrt{\dG{t+1}'} = \Pmat{{t+1}}.
            \end{align*}
            \item Since $\langle \dG{t},\1 \rangle = \muG{}{A_t}$, we get that the degree matrix of $\frac{1}{\muG{}{A_t}}\dG{t} \dG{t}'$ is $\U{t}$. Hence
            \[
            \mathcal{L} \left( \frac{1}{\muG{}{A_t}}\dG{t} \dG{t}'\right) = \U{t} - \frac{1}{\muG{}{A_t}}\dG{t} \dG{t}' = 
            \Uhalf (\I{t} - \frac{1}{\muG{}{A_t}}\sqrt{\dG{t}}\sqrt{\dG{t}'}) \Uhalf =
            \Uhalf \Pmat{t} \Uhalf.
            \]
    
            The result follows since $\Pmat{t}$ is $A_t\cap  \ter$-blocked (and has zeros on the diagonal outside of the  block).
    
            \item
            \begin{align*}
                v' \mathcal{L}\left(\frac{1}{\muG{}{A_t}} \dG{t} \dG{t}'\right) v &= 
                v'  (\U{t} - \frac{1}{\muG{}{A_t}}\dG{t} \dG{t}')  v =
                 v' \U{t} v  - \frac{1}{\muG{}{A_t}} v' \dG{t} \dG{t}' v \\&=
                 (\U{t}^{\frac{1}{2}} v)' (\U{t}^{\frac{1}{2}} v) - \frac{1}{\muG{}{A_t}} (v' \dG{t}) (\dG{t}' v) =
                 \norm{\U{t}^{\frac{1}{2}} v}_2^2 - \frac{1}{\muG{}{A_t}} \langle v,\dG{t} \rangle^2.
            \end{align*}
            \item Observe that $\sum_{i\in A_t \cap \ter}\norm{B(i)}_2^2 = \norm{\I{t} B}_F^2$. Indeed, $X= \I{t} B$ satisfies $\X{}{i,j} = B(i,j)$ if $i\in A_t \cap  \ter$ (and otherwise $\X{}{i,j} = 0$). Therefore
            \begin{align*}
                \sum_{i\in A_t \cap \ter}\norm{B(i)}_2^2 &= \norm{\I{t} B}_F^2 = \tr((\I{t} B)  (\I{t} B)') = 
                \tr (\I{t} B  B' \I{t}) \\&=
                \tr(\I{t}^2 B  B') = \tr(\I{t} B B').
            \end{align*}
            Where the fourth equality follows from Fact~\ref{fact:book_trace_identities} and the last equality follows from (2).
        \end{enumerate}
    \end{proof}
    
    We define the potential $\psi(t) = \tr[\W{t}^2] = \sum_{i\in A_t}{\norm{\W{t}{i}}_2^2}$, where $\W{t}$ was defined as $\W{t} = (\Pmat{t} \Dminushalf \F{t} \Dminushalf \Pmat{t})^\delta$. 
    Intuitively, by projecting using $\Pmat{t}$ , the potential only ``cares'' about the vertices of $A_t \cap \ter$. We prove in Lemma~\ref{lemma:F_expander} that having small potential certifies that $A_T$ is a near expander in $\F{t}$. 
    
    Before we bound the decrease in potential, we recall Definition~\ref{def:laplacian_normalized_laplacian} of a normalized Laplacian $\mathcal{N}(A) = \normalized{\mathcal{L}(A)} = I_{\ter} - \normalized{A}$, where $A$ is a symmetric $\ter$-blocked $\dG$-stochastic matrix. 
    The proof of the following technical lemma appears in Appendix~\ref{appendix-proofs}.

    \begin{lemma}
    \label{lemma:X_as_normalized_laplacian}
        For any matrix $A\in \RR^{n\times n}$, $\tr(A'(I_{\ter}-(\normalized{\N{t}})^{4\delta})A)\ge \frac{1}{3}\tr(A' \mathcal{N}(\M{t})A)$.
    \end{lemma}

    In the following lemma we bound the decrease in potential. The bound consists of the contribution of the matched vertices and it also takes into account the removal of $S_t$ from $A_t$.
    \begin{lemma}
    \label{lemma:potential_step_1}
        For each round $t$, \[\psi(t) - \psi(t+1) \ge \frac{1}{3}\sum_{\{i,k\}\in \M{t}} w_{ik}\norm{\left(\frac{\W{t}{i}}{\sqrt{\dG{}{i}}} - \frac{\W{t}{k}}{\sqrt{\dG{}{k}}}\right)}_2^2 + \sum_{j\in S_{t}\cap \ter} \dG{}{j} \norm{\frac{\W{t}{j}}{\sqrt{\dG{}{j}}}}_2^2\]
    \end{lemma}
    \begin{proof}
        To simplify the notation, we denote $\bar{N}_t \defeq \normalized{\N{t}}$ and $\bar{F}_t \defeq \normalized{\F{t}}$.
        We rewrite the potential in the next iteration as follows:
        \begin{align*}
            \psi(t+1) &= \tr(\W{{t+1}}^2) = \tr\left( \left( \Pmat{{t+1}} \Dminushalf \F{{t+1}} \Dminushalf \Pmat{{t+1}} \right)^{2\delta} \right) 
            \\
            &= \tr\left( \left( \Pmat{{t+1}} \Dminushalf (\N{{t}} U^{-1} \F{{t}} U^{-1} \N{{t}}) \Dminushalf \Pmat{{t+1}} \right)^{2\delta}\right) 
            \\
            &= \tr\left( \left( \Pmat{{t+1}} \Dminushalf (\N{{t}} \Dminushalf \Dminushalf \F{{t}} \Dminushalf \Dminushalf \N{{t}}) \Dminushalf \Pmat{{t+1}} \right)^{2\delta}\right) 
            \\
            &= \tr\left( \left( \Pmat{{t+1}} \bar{N}_{t} \bar{F}_t \bar{N}_{t}  \Pmat{{t+1}} \right)^{2\delta}\right) 
            \\
            &\numeq{6} \tr\left( \left(  \bar{N}_{t} \Pmat{{t+1}} \bar{F}_t  \Pmat{{t+1}} \bar{N}_{t}  \right)^{2\delta}\right) 
            \\
            &\numeq{7} \tr\left( \left(  \bar{N}_{t} \Pmat{{t+1}} \Pmat{t} \bar{F}_t \Pmat{t} \Pmat{{t+1}} \bar{N}_{t}   \right)^{2\delta}\right) 
            \\
            &= \tr\left( \left(  \bar{N}_{t} \Pmat{{t+1}} (\Pmat{t} \bar{F}_t \Pmat{t}) \Pmat{{t+1}} \bar{N}_{t}   \right)^{2\delta}\right) \ ,
        \end{align*}
        where Equality $(6)$ follows from part $1$ of Lemma~\ref{lemma:submatrix_properties} for $\N{{t}}$ (which is $\A{t+1}$-blocked $\dG$-stochastic by Lemma~\ref{lemma:basic_properties}), and Equality $(7)$ follows from part $3$ of Lemma~\ref{lemma:submatrix_properties}.

        By Properties (1) and (2) of Lemma~\ref{lemma:submatrix_properties} it holds that $\bar{N}_{t} \Pmat{{t+1}} = \Pmat{{t+1}} \bar{N}_{t}  = \Pmat{{t+1}} \bar{N}_{t} \Pmat{{t+1}}$.  
        Therefore, the potential can be written in terms of symmetric matrices:
        \begin{align*}
            \psi(t+1) &= \tr\left( \left(  (\Pmat{{t+1}} \bar{N}_{t} \Pmat{{t+1}}) (\Pmat{t} \bar{F}_t \Pmat{t}) (\Pmat{{t+1}} \bar{N}_{t} \Pmat{{t+1}})   \right)^{2\delta}\right) 
            \\
            &\le \tr ((\Pmat{{t+1}} \bar{N}_{t} \Pmat{{t+1}})^{2\delta} (\Pmat{t} \bar{F}_t \Pmat{t})^{2\delta} (\Pmat{{t+1}} \bar{N}_{t} \Pmat{{t+1}})^{2\delta}) 
            \\
            &\numeq{2} \tr ( (\Pmat{{t+1}} \bar{N}_{t} \Pmat{{t+1}})^{4\delta} (\Pmat{t} \bar{F}_t \Pmat{t})^{2\delta}) = \tr((\bar{N}_{t} \Pmat{{t+1}})^{4\delta} \W{t}^2) 
            \\
            &\numeq{4} \tr(\bar{N}_{t}^{4\delta}\Pmat{{t+1}} \W{t}^2) \numeq{5} \tr(\bar{N}_{t}^{2\delta}\Pmat{{t+1}}\bar{N}_{t}^{2\delta} \W{t}^2) 
            \\
            &\numeq{6} \tr(\W{t} \bar{N}_{t}^{2\delta}\Pmat{{t+1}}\bar{N}_{t}^{2\delta} \W{t})
            \\
            &\numeq{7} \tr(\W{t} \bar{N}_{t}^{2\delta} \Dminushalf \mathcal{L}\left(\frac{1}{\muG{}{A_{t+1}}} \dG{t+1} \dG{t+1}'\right) \Dminushalf \bar{N}_{t}^{2\delta} \W{t}) 
            \\
            &= \tr\left(\left(\Dminushalf \cdot \bar{N}_{t}^{2\delta} \W{t} \right)' \cdot \mathcal{L}\left(\frac{1}{\muG{}{A_{t+1}}} \dG{t+1} \dG{t+1}'\right) \cdot \left(\Dminushalf \cdot \bar{N}_{t}^{2\delta} \W{t}\right)\right)\ ,
        \end{align*}
        where the inequality follows from Theorem~\ref{theorem:symmetric_rearrangement}, Equality $(2)$ follows from Fact~\ref{fact:book_trace_identities}, Equalities $(4)$  and $(5)$ follow from Properties (1) and (2) of Lemma~\ref{lemma:submatrix_properties} (and from the fact that $\N{{t}}$ is $\A{t+1}$-blocked $\dG$-stochastic, by Lemma~\ref{lemma:basic_properties}), Equality $(6)$ again uses Fact~\ref{fact:book_trace_identities} and Equality $(7)$ follows from part $4$ of Lemma~\ref{lemma:submatrix_properties}.
        
        Let $\Z{t} = \Dminushalf \cdot \bar{N}_{t}^{2\delta} \W{t}$.  By applying part $5$ of Lemma~\ref{lemma:submatrix_properties} we get
        \begin{align}
            \psi(t+1) &\le \tr\left(\Z{t}' \mathcal{L}\left(\frac{1}{\muG{}{A_{t+1}}} \dG{t+1} \dG{t+1}'\right) \Z{t}\right) = \sum_{i=1}^n (\Z{t}{,i})' \mathcal{L}\left(\frac{1}{\muG{}{A_{t+1}}} \dG{t+1} \dG{t+1}'\right) \Z{t} (,i) 
            \notag \\
            &\numeq{2} \sum_{i=1}^n \left(\norm{\U{{t+1}}^{\frac{1}{2}} \Z{t}{,i}}_2^2 - \frac{1}{\muG{}{A_{t+1}}} \left\langle \Z{t}{,i}, \dG{t+1}\right\rangle^2\right) \le \sum_{i=1}^n \norm{\U{{t+1}}^{\frac{1}{2}} \Z{t}{,i}}_2^2 
            \notag \\
            &= \sum_{i=1}^n \sum_{j\in A_{t+1}} \left(\sqrt{\dG{}{j}} \Z{t}{j,i}\right)^2 = \sum_{j\in A_{t+1}} \norm{\Matrix{\U{{t+1}}^{\frac{1}{2}}\Z{t}}{j}}_2^2 
            \numeq{5} \sum_{j\in A_{t+1}} \norm{\Matrix{\bar{N}_{t}^{2\delta}\W{t}}{j}}_2^2 
            \notag \\
            &= \sum_{j\in A_{t}} \norm{\Matrix{\bar{N}_{t}^{2\delta}\W{t}}{j}}_2^2 - \sum_{j\in S_{t}} \norm{\Matrix{\bar{N}_{t}^{2\delta}\W{t}}{j}}_2^2, \label{eq:psit+1ep}
        \end{align}
        where Equality $(2)$ holds by Property (5) of Lemma~\ref{lemma:submatrix_properties} and Equality $(5)$ holds since we only sum rows in $A_{t+1}$.
        Since $\bar{N}_{t}$ is diagonal outside $A_{t+1}$ (By the definition of $\M{t}$),
        we have that $\Matrix{\bar{N}_{t}^{2\delta} \W{t}}{j} = \W{t}{j}$, for every $j\in S_t$.
        Thus,
        \begin{equation}\label{eq:n(N)_W_correlation}
            \sum_{j\in S_{t}} \norm{\Matrix{\bar{N}_{t}^{2\delta}\W{t}}{j}}_2^2 =
        \sum_{j\in S_{t}} \norm{\W{t}{j}}_2^2. 
        \end{equation}
        
        By part $6$ of Lemma~\ref{lemma:submatrix_properties}, we get
        \begin{align}
            \sum_{j\in A_{t}} \norm{\Matrix{\bar{N}_{t}^{2\delta}\W{t}}{j}}_2^2 &= 
            \sum_{j\in A_{t}\cap \ter} \norm{\Matrix{\bar{N}_{t}^{2\delta}\W{t}}{j}}_2^2 \\&=
            \tr(\I{{t}} \cdot \bar{N}_{t}^{2\delta} \cdot \W{t}^2 \cdot \bar{N}_{t}^{2\delta}) \notag \\&=
            \tr(\bar{N}_{t}^{2\delta} \cdot \I{{t}} \cdot  \W{t}^2 \cdot \bar{N}_{t}^{2\delta}) \notag \\&=
            \tr(\bar{N}_{t}^{2\delta} \cdot  \W{t}^2 \cdot \bar{N}_{t}^{2\delta}) \notag \\&=
            \tr(   \bar{N}_{t}^{4\delta}\W{t}^2 ) \label{eq:At}
        \end{align}
        
        where the first equality holds because $\N{t}(i)=0$  for $i\in \ter$. The second equality holds by part $6$ of Lemma~\ref{lemma:submatrix_properties}. The third equality holds since $\N{{t}}$ is $\A{t+1}\cap \ter$-blocked $\dG$-stochastic (by Lemma~\ref{lemma:basic_properties}), so in particular it is $\A{t}\cap \ter$-blocked $\dG{t}$-stochastic, and we can use part $1$ of Lemma~\ref{lemma:submatrix_properties}. The fourth equality holds because $\I{t} \W{t}  = \I{t} (\Pmat{t} \bar{F}_{t} \Pmat{t})^\delta$ 
        and $\I{t} \Pmat{t}  = \Pmat{t}$ (by part $2$ of Lemma~\ref{lemma:submatrix_properties}), and the last equality follows from Fact~\ref{fact:book_trace_identities}. 
        Plugging 
        Equations (\ref{eq:n(N)_W_correlation}) and (\ref{eq:At}) 
        into (\ref{eq:psit+1ep}) we get the following bound on the decrease in potential: 
        \begin{align*}
            \psi(t) - \psi(t+1) &\ge \tr(   (I-\bar{N}_{t}^{4\delta})\W{t}^2 ) + 
            \sum_{j\in S_{t}} \norm{\W{t}{j}}_2^2 \\&=
            \tr( \W{t}  (I_{\ter}-\bar{N}_{t}^{4\delta})\W{t} ) + 
            \sum_{j\in S_{t}} \norm{\W{t}{j}}_2^2 \\&\ge
            \frac{1}{3}\tr( \W{t}  \mathcal{N}(\M{{t}}) \W{t} ) + 
            \sum_{j\in S_{t}} \norm{\W{t}{j}}_2^2  \\&=
            \frac{1}{3}\tr( (\Dminushalf \W{t})'  \mathcal{L}(\M{{t}}) (\Dminushalf \W{t}) ) + 
            \sum_{j\in S_{t}\cap  \ter}  \dG{}{j} \norm{\frac{\W{t}{j}}{\sqrt{\dG{}{j}}}}_2^2 \\&=
            \frac{1}{3}\sum_{\{i,k\}\in \M{t}}{w_{ik}\norm{\frac{\W{t}{i}}{\sqrt{\dG{}{i}}} - \frac{\W{t}{k}}{\sqrt{\dG{}{k}}}}_2^2} + \sum_{j\in S_{t} \cap \ter} \dG{}{j} \norm{\frac{\W{t}{j}}{\sqrt{\dG{}{j}}}}_2^2
        \end{align*}
        
        where the first equality follows from Fact~\ref{fact:book_trace_identities} and since as implied before $I_{\ter}\W{t}=\W{t}$, the second inequality follows Lemma~\ref{lemma:X_as_normalized_laplacian},  and the last equality follows from Lemma~\ref{lemma:matching_laplacian}.
    \end{proof}

    The following shows that the potential is expected to drop by a factor of $1-\Omega(1/\log n)$.
    \begin{lemma}
    \label{lemma:potential_step_2}
        For each round $t$, 
        \[\EE\left[\frac{1}{3}\sum_{\{i,k\}\in \M{t}} w_{ik} \norm{\frac{\W{t}{i}}{\sqrt{\dG{}{i}}} - \frac{\W{t}{k}}{\sqrt{\dG{}{k}}}}_2^2 + \sum_{j\in S_{t} \cap \ter} \dG{}{j} \norm{\frac{\W{t}{j}}{\sqrt{\dG{}{j}}}}_2^2\right]  
        \ge
        \frac{1}{3000 \alpha \log n}\psi(t) - \frac{1}{n^{\alpha/16}}\]
        for some constant $\alpha$, where the expectation is taken over the unit vector $r\in \RR^n$.
    \end{lemma}    
    \begin{proof}
        Recall that $u_i = \frac{1}{\sqrt{\dG{}{i}}} \langle \W{t}{i}, r \rangle$ for $i\in A_t \cap \ter$. Notice that $\Matrix{\frac{\W{t}{i}}{\sqrt{\dG{}{i}}} }{j} = \frac{\W{t}{i,j}}{\sqrt{\dG{}{i}}}$. 
        Use Lemma~\ref{lemma:projection_pairs} from Appendix~\ref{appendix:projection} for the set of vectors $\left\{\frac{1}{\sqrt{\dG{}{i}}}\W{t}{i} \mid i\in A_t \cap \ter \right\}$. 
        By Lemma~\ref{lemma:projection_pairs}(2), we have with high probability:
        \begin{align*}
            \forall i, k\in A_t \cap \ter &: \norm{\frac{\W{t}{i}}{\sqrt{\dG{}{i}}} - \frac{\W{t}{k}}{\sqrt{\dG{}{k}}}}_2^2 \ge \frac{n}{\alpha\log n}\cdot(u_i-u_k)^2
            \\
            \forall i\in A_t \cap  \ter&: \norm{\frac{\W{t}{i}}{\sqrt{\dG{}{i}}}}_2^2 \ge \frac{n}{\alpha\log n}\cdot u_i^2
        \end{align*}
        for some constant $\alpha > 0$. In order to replace the inequality with high probability by an inequality in expected values, we introduce a random variable $z$ that is non-zero only when this inequality fails to hold, such that  
        \begin{align*}
	        \forall i, k\in A_t \cap \ter &: \norm{\frac{\W{t}{i}}{\sqrt{\dG{}{i}}} - \frac{\W{t}{k}}{\sqrt{\dG{}{k}}}}_2^2 \ge \frac{n}{\alpha \log n}\cdot(u_i - u_k)^2 - z
	        \\
            \forall i\in A_t \cap \ter &: \norm{\frac{\W{t}{i}}{\sqrt{\dG{}{i}}}}_2^2 \ge \frac{n}{\alpha\log n}\cdot u_i^2 - z
        \end{align*}
        holds with probability $1$. \textit{I.e.}, we define
        \begin{align*}
            \mathcal{B} = \{0\} &\cup \left\{\frac{n}{\alpha\log n}(u_i - u_k)^2-\norm{\frac{\W{t}{i}}{\sqrt{\dG{}{i}}} - \frac{\W{t}{k}}{\sqrt{\dG{}{k}}}}_2^2 : (i, k)\in (A_t \cap \ter )\times (A_t \cap \ter)\right\}
            \\
            &\cup \left\{\frac{n}{\alpha\log n}u_i^2-\norm{\frac{\W{t}{i}}{\sqrt{\dG{}{i}}}}_2^2 : i\in A_t \cap \ter\right\}
        \end{align*}
        and $z = \max(\mathcal{B})$. Let $\muBound \coloneqq \langle \1, \muG \rangle = \poly(n)$.
        We get that
        \begin{align*}
             \frac{1}{3}\sum_{\{i,k\}\in \M{t}} w_{ik}\norm{\frac{\W{t}{i}}{\sqrt{\dG{}{i}}} - \frac{\W{t}{k}}{\sqrt{\dG{}{k}}}}_2^2 &\ge
             \frac{n}{3\alpha \log{n}}\sum_{\{i,k\}\in \M{t}} w_{ik}(u_i - u_k)^2 - \muBound\cdot z
             \\
             \sum_{j\in S_{t}} \dG{}{j} \norm{\frac{\W{t}{j}}{\sqrt{\dG{}{j}}}}_2^2 &\ge \frac{n}{\alpha \log{n}} \sum_{j\in S_{t} \cap \ter} \dG{}{j} u_j^2 - K\cdot z.
        \end{align*}
        
        This means that 
        \begin{align*}
            &\frac{1}{3}\sum_{\{i,k\}\in \M{t}} w_{ik}\norm{\left(\frac{\W{t}{i}}{\sqrt{\dG{}{i}}} - \frac{\W{t}{k}}{\sqrt{\dG{}{k}}}\right)}_2^2 + \sum_{j\in S_{t}\cap \ter} \dG{}{j} \norm{\frac{\W{t}{j}}{\sqrt{\dG{}{j}}}}_2^2
            \\
            &\ge \frac{n}{3\alpha \log{n}}\sum_{\{i,k\}\in \M{t}} w_{ik}(u_i - u_k)^2 + 
            \frac{n}{\alpha \log{n}} \sum_{j\in S_{t} \cap \ter} \dG{}{j} u_j^2  - 2\muBound\cdot z\\&\numge{2}
            \frac{n}{3\alpha \log{n}}\sum_{i\in A_t^{l}\setminus S_t} m_t(i) (u_i - \eta)^2 + 
            \frac{n}{\alpha \log{n}} \sum_{j\in S_{t}} \dG{}{j} u_j^2  - 2\muBound\cdot z\\&\ge
            \frac{n}{3\alpha \log{n}}\sum_{i\in A_t^{l}\setminus S_t} m_t(i) (u_i - \eta)^2 + 
            \frac{n}{\alpha \log{n}} \sum_{j\in A_t^{l} \cap S_{t}} \dG{}{j} u_j^2  - 2\muBound\cdot z\\&\numge{4}
            \frac{n}{27\alpha \log{n}}\sum_{i\in A_t^{l}\setminus S_t} m_t(i) u_i^2 + 
            \frac{n}{\alpha \log{n}} \sum_{j\in A_t^{l} \cap S_{t}} \dG{}{j} u_j^2  - 2\muBound\cdot z\\&\numge{5}
            \frac{n}{27\alpha \log{n}}\sum_{i\in A_t^{l}} m_t(i) u_i^2 - 2\muBound\cdot z \ge
            \frac{n}{3000\alpha \log{n}}\sum_{i\in A_t} \dG{}{i} u_i^2 - 2\muBound\cdot z \ .
        \end{align*} 
        Inequality $(2)$ is due to Lemma~\ref{lemma:RST}(1), the fact that $\muG(v)=0$ for every $v\notin \ter$ and because each $i\in A^l_t \setminus S_t$ is matched to $A^r_t\setminus S_t$ with a total weight of $m_i$ (that is, $\sum_{j:\{i,j \}\in \M{t}}w_{ij}= m(i)$ for $i\in A_t^{l}\setminus S_t$), Inequality $(4)$ follows from Lemma~\ref{lemma:RST}(4), Inequality $(5)$ is true because $m_i\le \dG{}{i}$ for all $i\in A^l_t$, and the last inequality follows from Lemma~\ref{lemma:RST}(5).
        Finally, in expectation,
        \begin{align*}
            &\EE\left[\frac{1}{3}\sum_{\{i,k\}\in \M{t}} w_{ik}\norm{\left(\frac{\W{t}{i}}{\sqrt{\dG{}{i}}} - \frac{\W{t}{k}}{\sqrt{\dG{}{k}}}\right)}_2^2 + \sum_{j\in S_{t}\cap \ter} \dG{}{j} \norm{\frac{\W{t}{j}}{\sqrt{\dG{}{j}}}}_2^2\right]
            \\
            &\ge \frac{n}{3000\alpha \log{n}}\sum_{i\in A_t} \dG{}{i} \EE[u_i^2] - 2\muBound\cdot\EE[z]\\&=
            \frac{1}{3000\alpha \log{n}}\sum_{i\in A_t\cap \ter} \dG{}{i} \norm{\frac{\W{t}{i}}{\sqrt{\dG{}{i}}}}_2^2 - 2\muBound\cdot\EE[z] \\&=
            \frac{1}{3000\alpha \log{n}}\sum_{i\in A_t} \norm{\W{t}{i}}_2^2 - 2\muBound\cdot\EE[z] \\&=
            \frac{1}{3000\alpha \log{n}}\psi(t) - 2\muBound\cdot\EE[z] 
        \end{align*}
        where the first equality follows from Lemma~\ref{lemma:projection_pairs} (1).
        
        We note that $z = O(\poly(n))$. Indeed, for every $i \in \A{t} \cap \ter$
        \[
        u_i = \left\langle \frac{\W{t}{i}}{\sqrt{\dG{}{i}}}, r \right\rangle \le
        \norm{\frac{\W{t}{i}}{\sqrt{\dG{}{i}}}}_2 = 
        \frac{\norm{\W{t}{i}}_2}{\sqrt{\dG{}{i}}} \le \frac{1}{\sqrt{\dG{}{i}}} = O(\poly(n)),\footnote{The last equality explains why Assumption~\ref{assumption:mu} requires $\min_{v\in \ter}\mu(v) \ge \frac{1}{\poly(n)}$.}
        \]
        where the first inequality holds due to Cauchy-Schwartz since $r$ is a unit vector and the last inequality holds since all eigenvalues of $\W{t}$ are in $[0,1]$: By Lemma~\ref{lemma:normalized-laplacian-eigenvalues} we get that $\Dminushalf \F{t} \Dminushalf = I_{\ter} - \mathcal{N}(\F{t})$ has eigenvalues in $[-1,1]$. Since $\Pmat{t}$ is a projection matrix, it has eigenvalues in $[0,1]$. Therefore, $\Pmat{t} \Dminushalf \F{t} \Dminushalf \Pmat{t}$ has eigenvalues in $[-1,1]$. Finally, $\W{t}= (\Pmat{t} \Dminushalf \F{t} \Dminushalf \Pmat{t})^\delta$ is PSD since $\delta$ is a power of $2$, so its eigenvalues are in $[0,1]$.
        
        By  Lemma~\ref{lemma:projection_pairs}, $z$ is non-zero with probability at most $\frac{1}{n^{\alpha/8}}$, so for large enough $\alpha$,  we get
        $2\muBound \EE[z] \le \frac{2\muBound \poly(n)}{n^{\alpha/8}} \le \frac{1}{n^{\alpha/16}}$,\footnote{The second Inequality explains why Assumption~\ref{assumption:mu} requires $\max_{v\in \ter}\mu(v) \le \poly(n)$.} completing the proof. 
    \end{proof}

    The following two corollaries follow by Lemmas~\ref{lemma:potential_step_1} and~\ref{lemma:potential_step_2}.
    
\begin{corollary}
    \label{cor:potential_reduction}
        For each round $t$, $\EE[\psi(t+1)] \le
        \left(1-\frac{1}{3000 \alpha \log n}\right)\psi(t) + \frac{1}{n^{\alpha/16}}$, where the expectation is over the unit vector $r\in \RR^n$.
    \end{corollary}
    
    \begin{corollary} [Total Potential]
    \label{cor:total_potential}
        With high probability over the choices of $r$, $\psi(T) \le \frac{1}{n}$.
    \end{corollary}

    \begin{lemma}
    \label{lemma:F_expander}
    Let $H=(V,\bar{E})$ be a graph on
    $n$ vertices, such that $\F{T}$ is its  weighted adjacency matrix. Assume that $\psi(T) \le \frac{1}{n}$. Then, $A_T$ is a near $(\frac{1}{5},\muG)$-expander with respect to $\mu$ in $H$.
    \end{lemma}
    \begin{proof}
        Recall that $\F{T}$ is symmetric and $\dG{T}$-stochastic. Let $k = \muG{}{A_T}$.
        Let $S\subseteq A_T$ be a cut, and denote $\dG{S}\in\RR^n$ to be the vector where $\dG{S}{u} =
        \left\{\begin{array}{cl}
            \dG{}{u} & \mbox{if $u\in S$,} \\
            0 &\mbox{otherwise.}
            \end{array}\right.$ 
        Additionally, denote $\ell = \muG{}{S} \le \frac{1}{2}k$. Note that $\norm{\sqrt{\dG{S}}}_2^2 = \ell$. 
         
        Denote by $\bar{\lambda} \ge 0$ the largest singular value of $\X{T} \defeq \Pmat{T} \Dminushalf \F{T} \Dminushalf \Pmat{T}$ (square root of the largest eigenvalue of $(\Pmat{T} \Dminushalf \F{T} \Dminushalf \Pmat{T})^2$).
        Because $\tr(\X{T}^{2 \delta}) = \psi(T) \le \frac{1}{n}$, we have in particular that the largest eigenvalue of $\X{T}^{2 \delta}$ is at most $\frac{1}{n}$, so we have $\bar{\lambda} \le \frac{1}{n^{\frac{1}{2\delta}}}$. We choose $\delta  = \Theta(\log n)$\footnote{This gives the upper bound on $\delta$. The lower bound comes from Lemma~\ref{lemma:F_t_embeddable_in_G_t}, ensuring a low congestion embedding.} such that $\frac{1}{n^{\frac{1}{2\delta}}}\le \frac{1}{20}$, so  $\bar{\lambda}\le\frac{1}{20}$. 
            
        In order to prove near-expansion we need to lower bound $|E_{\F{T}}(S,\comp{S})|$. We do so by upper bounding $|E_{\F{T}}(S,S)| =  \1_S' \F{T} \1_S$. Because $\F{t}$ is $\ter$-blocked we get that $\1_S' \F{T} \1_S = \1_S' (\I{T} \F{T} \I{T}) \1_S$. Observe the following relation between $\X{T}$ and $\I{T} \F{T} \I{T}$: 
        \begin{align*}
            \Uhalf\X{T}\Uhalf &= \Uhalf(\Pmat{T} \Dminushalf \F{T} \Dminushalf \Pmat{T}) \Uhalf 
            =
            \Uhalf(\I{T} - \frac{1}{k}\sqrt{\dG{T}}\sqrt{\dG{T}'}) \Dminushalf \F{T} \Dminushalf 
            (\I{T} - \frac{1}{k}\sqrt{\dG{T}}\sqrt{\dG{T}'}) \Uhalf 
            \\
            &= (\I{T} - \frac{1}{k}\dG{T}\1_T') \F{T} (\I{T} - \frac{1}{k}\1_T\dG{T}') 
            =
            \I{T} \F{T} \I{T} - \frac{1}{k}\dG{T}\1_T' \F{T} \I{T} - \frac{1}{k} \I{T} \F{T} \1_T\dG{T}' + 
            \frac{1}{k^2} \dG{T}\1_T' \F{T} \1_T \dG{T}'.
        \end{align*}

        Rearranging the terms, we get
        \begin{align*}
            \I{T} \F{T} \I{T} = 
            \Uhalf \X{T} \Uhalf 
            +\frac{1}{k}\dG{T}\1_T' \F{T} \I{T}
            +\frac{1}{k} \I{T} \F{T} \1_T\dG{T}'
            -\frac{1}{k^2} \dG{T}\1_T' \F{T} \1_T \dG{T}'  \ .
        \end{align*}
        
        Therefore
         \begin{align*}
            |E_{\F{T}}(S,S)| = & \1_S' \F{T} \1_S =
            \1_S' \left(
            \Uhalf\X{T} \Uhalf +
            \frac{1}{k}\dG{T}\1_T' \F{T} \I{T} + 
            \frac{1}{k} \I{T} \F{T} \1_T\dG{T}' - 
            \frac{1}{k^2} \dG{T}\1_T' \F{T} \1_T \dG{T}'
            \right) \1_S .
        \end{align*}
        
        We analyze the summands separately. The first summand can be bounded using $\bar{\lambda}$, the largest singular value of $\X{T}$:
         \begin{align*}
            \1_S' \Uhalf\X{T}\Uhalf \1_S = \sqrt{\dG{S}'} X \sqrt{\dG{S}} =
            \left\langle \sqrt{\dG{S}}, X \sqrt{\dG{S}} \right\rangle \le
            \norm{\sqrt{\dG{S}}}_2 \norm{\X{T}\sqrt{\dG{S}}}_2 \le \norm{\sqrt{\dG{S}}}_2^2 \bar{\lambda} \le \frac{\ell}{20},
        \end{align*}
          where the first inequality is the Cauchy-Schwartz inequality. Observe that the second and third summands are equal because:
        \begin{align*}
            \frac{1}{k} \1_S' \dG{T}\1_T' \F{T} \I{T} \1_S =
            \frac{\ell}{k} \1_T' \F{T} \1_S = 
            \frac{\ell}{k} \1_S' \F{T} \1_T = 
            \frac{1}{k} \1_S' \I{T} \F{T} \1_T \dG{T}' \1_S,
        \end{align*}
        
        where the second equality follows by transposing and since $\F{T}$ is symmetric. We now bound the sum of the second, third and fourth summands:
        \begin{align*}
            &\1_S' \left(
            \frac{1}{k}\dG{T}\1_T' \F{T} \I{T} + 
            \frac{1}{k} \I{T} \F{T} \1_T\dG{T}' - 
            \frac{1}{k^2} \dG{T}\1_T' \F{T} \1_T \dG{T}'\right) \1_S =
            \frac{2\ell}{k} \1_T' \F{T} \1_S - \frac{\ell^2}{k^2}\1_T' \F{T} \1_T \\&
            \numle{1}
            \left( \frac{2\ell}{k} - \frac{\ell^2}{k^2} \right)\1_T' \F{T} \1_S \le
            \left( \frac{2\ell}{k} - \frac{\ell^2}{k^2} \right)\1' \F{T} \1_S = 
            \left( \frac{2\ell}{k} - \frac{\ell^2}{k^2} \right)\dG' \1_S = 
            \frac{\ell}{k}\left( 2 - \frac{\ell}{k} \right) \ell,
        \end{align*}
        
        where the Inequality~$(1)$ follows since $S \subseteq \A{t}$. Note that $\frac{\ell}{k}\in [0,\frac{1}{2}]$. The last inequality is true because for $\frac{\ell}{k}$ in this range, $\left(\frac{2\ell}{k}-\frac{\ell^2}{k^2}\right)\ge 0$. Moreover, because $\frac{\ell}{k}\in \left[0,\frac{1}{2}\right]$, we have $\frac{\ell}{k}\left(2-\frac{\ell}{k}\right) \le \frac{3}{4}$. Therefore, $|E_{\F{T}}(S,S)| \le \frac{1}{20}\ell + \frac{3}{4}\ell = \frac{4}{5}\ell$, and
        \begin{align*}
            |E(S,\comp{S})| &= \sum_{u\in S}{\sum_{v\in \comp{S}}{\F{T}{u,v}}} = \sum_{u\in S}{\sum_{v\in V}{\F{T}{u,v}}} - \sum_{u\in S}{\sum_{v\in S}{\F{T}{u,v}}}
            \\
            &= \sum_{u\in S}{\dG{}{u}} - \sum_{u\in S}{\sum_{v\in S}{\F{T}{u,v}}} \ge \ell - \frac{4}{5}\ell = \frac{\ell}{5} \ .
        \end{align*}
        So $\Phi_G(S, \comp{S}) = \frac{|E(S,\comp{S})|}{\muG{}{S}} \ge \frac{1}{5}$, and this is true for all cuts $S\subseteq A$ with $\frac{\muG{}{S}}{\muG{}{A_t}}\le\frac{1}{2}$.
            
    \end{proof}

    \begin{corollary}
    \label{cor:G_expander}
        If we reach round $T$, then with high probability, $A_T$ is a near $(\phi, \muG)$-expander in $G$.
    \end{corollary}
    \begin{proof}
        Assume we reach round $T$. By Corollary~\ref{cor:total_potential} and Lemma~\ref{lemma:F_expander}, with high probability, $A_T$ is a near $(\frac{1}{5},\muG)$-expander in $\F{T}$. By Lemma~\ref{lemma:F_t_embeddable_in_G_t}, $\F{T}$ is embeddable in $G_T$ with congestion $\frac{4}{\delta}$. 
        Therefore, by Lemma~\ref{lemma:near_expansion_and_embedding}, $A_T$ is a near $(\frac{\delta}{10},\muG)$-expander in $G_T$.
        
        Furthermore, by Lemma~\ref{lemma:G_t_embeddable_in_G}, $G_T$ is embeddable in $G$ with congestion $2cT$. Applying Lemma~\ref{lemma:near_expansion_and_embedding} again, we get that $A_T$ is a near $(\frac{\delta}{10cT},\muG)$-expander in $G$. After fixing the constants in $\delta=\Theta(\log n)$ and $T=\Theta(\log^2 n)$, we choose the fixed constant $c_1>0$ sufficiently small that $\frac{\delta\log n}{10c_1T}\ge1$ for all sufficiently large $n$. With $c=c_1/(\phi\log n)$, this gives $\frac{\delta}{10cT}\ge\phi$, as required.
    \end{proof}

    \subsection{Proof of Theorem \texorpdfstring{\ref{theorem:main-result}}{3.1}}
    \label{section:main-result_theorem_proof}
    
    We are now ready to prove Theorem~\ref{theorem:main-result}. 
    \begin{proof} [Proof of Theorem~\ref{theorem:main-result}]
        
        Recall that $S_t$ denotes the cut returned by Lemma~\ref{lemma:fair-cuts} (applied on $G[A_t]$) at iteration $t$, so that $A_{t+1} = A_{t}\setminus S_{t}$. Let $c_0 \defeq \frac{7}{c_1}$.
	    
	    Observe first that in any round $t$, we have $\Phi^{\muG}_{G}(A_t, R_t)\le\frac{7}{c}=O(\phi\log n)$. 
	    Indeed, since $R_t=\bigcup_{0\le t' < t}{S_{t'}}$, Lemma~\ref{lemma:fair-cuts} implies that, with high probability, for each $t'$, $\Phi^{\muG}_{G[A_{t'}]}(S_{t'}, \bar{S_{t'}})\le\frac{7}{c}=O(\phi\log n)$. 
	    
	    Assume Algorithm~\ref{algo:cut_matching} terminates because $\muG{}{R_t}>\frac{\muG{}{V}\cdot c\cdot \phi}{70}=\Omega(\frac{\muG{}{V}}{\log n})$.
        Hence, $\muG{}{R_{t-1}}\le \frac{\muG{}{V}\cdot c\cdot \phi}{70} = \frac{\muG{}{V}\cdot c_1}{70\log n}$.
        By Lemma~\ref{lemma:fair-cuts}, we also have $\muG{}{S_{t-1}} \le \frac{2}{3} \muG{}{A_{t-1}} \le \frac{2}{3}\muG{}{V}$ with high probability. 
        Combining these inequalities gives $\muG{}{R_t}=\muG{}{R_{t-1}}+\muG{}{S_{t-1}} \le \frac{3}{4}\muG{}{V}$\footnote{For sufficiently large $n$, we have $\log n \ge c_1$.}. 
        In particular, $\muG{}{A_t}= \muG{}{V} -\muG{}{R_t} = \Omega(\muG{}{V})= \Omega(\frac{\muG{}{V}}{\log n})$. Therefore $(A_t,R_t)$ is a balanced cut where $\Phi_{G}(A_t, R_t)=O(\phi\log n)$. We end in Case (2) of Theorem~\ref{theorem:main-result}.  
	    
	    Otherwise, Algorithm~\ref{algo:cut_matching} terminates at round $T$ and we apply Corollary~\ref{cor:G_expander}. If $R=\emptyset$, then we obtain the first case of Theorem~\ref{theorem:main-result} because the whole vertex set $V$ is, with high probability, a near $(\Omega(\phi),\muG)$-expander, which means that $G$ is an $(\Omega(\phi), \muG)$-expander. Otherwise, we have $\Phi_{G}(A_T, R_T)\le\frac{7}{c}=\frac{7}{c_1} \phi \log n = c_0 \phi\log n$.
	    Additionally, $\muG{}{R_{t}}\le\frac{\muG{}{V}\cdot c\cdot \phi}{70}=\frac{\muG{}{V}\cdot c_1}{70\log n}=\frac{\muG{}{V}}{10c_0\log n}$, and, with high probability, $A_T$ is a near $(\Omega(\phi),\mu)$-expander in $G$, which means we obtain the third case of Theorem~\ref{theorem:main-result}.

	    Finally, the running time of Algorithm ~\ref{algo:cut_matching} is $\tilde{O}(m)$:
        The time of iteration $t$ is the sum of the running times of the following steps:
        \begin{enumerate}
            \item Sample a random unit vector $r\in \RR^n$.
            \item Compute the projections vector $u = \Dminushalf \W{t} \cdot r$. This takes $O(t\cdot\delta\cdot m)=O(m\cdot t\cdot\log n)$ time since $\W{t}$ is a multiplication of $O(t \cdot \delta)$ matrices, where each matrix either has $O(m)$ non-zero entries or is a projection matrix $\Pmat{t}$.
            \item Computing $A^l_t$ and $A^r_t$ in time $O(n\log n)$ (Lemma~\ref{lemma:RST}).
            \item Computing the cut $S_t$ and the flow $f$ on $G[V\setminus S_t]$ in time $\tilde{O}(m)$ (Lemma~\ref{lemma:fair-cuts}).
            \item Moving $S_t$ from $A_t$ to $R_{t+1}$ in time $O(m)$.
            \item Constructing $M_{t}$ in time $O(m\log n)$ (using dynamic trees \cite{ST83}).
        \end{enumerate}
        This gives a total running time of $\tilde{O}(m)$ for iteration $t$. As $t$ ranges from $1$ to $T=\Theta(\log^2 n)$, this completes the proof of Theorem~\ref{theorem:main-result}.
    \end{proof}

\section{Expander Decomposition}\label{section:expander_decomposition}

    In this section, we present the standard derivation of an expander decomposition (Theorem~\ref{theorem:expander_decomposition}) from Theorem~\ref{theorem:cut_matching}. 
    
    The following lemma, called ``the trimming step'', is the last key component for the expander decomposition algorithm. This is a slight generalization of~\cite[Theorem 8.2]{LNPSsoda13}. For completeness, we include its proof in Appendix~\ref{appendix-proofs}.

    \begin{theorem}
    \label{theorem:trimming}
        Given a graph $G=(V, E)$, a set $A\subseteq V$, a parameter $\phi>0$, and a vertex measure $\muG : V \to \RR_{\ge 0}$, such that $A$ is a near $(\phi,\muG)$-expander in $G$ and $|E(A, \comp{A})|\le\frac{\phi\cdot\mu(A)}{9}$,
        there exists a randomized algorithm (the trimming step) which finds in time $\tilde{O}(m)$ a set $A'\subseteq A$ such that with high probability:
        \begin{itemize}
            \item $\Phi^\mu(G[A'])\ge\phi/6$,
            \item $\mu(A')\ge\mu(A) - \frac{4|E(A, \comp{A})|}{\phi}$, and
            \item $|E(A', \comp{A'})|\le 2|E(A, \comp{A})|$.
        \end{itemize}
    \end{theorem}

    Using Theorem~\ref{theorem:trimming} and Theorem~\ref{theorem:cut_matching}, we prove the following improvement on Theorem~\ref{theorem:cut_matching}.
    
    \begin{theorem}
    \label{theorem:main-result3}
        Given a graph $G=(V,E)$ with $n$ vertices and $m$ edges, a parameter $\phi > 0$, and a vertex measure $\mu : V \to \RR_{\ge 0}$ satisfying Assumption \ref{assumption:mu}, there exists a randomized algorithm which takes $\tilde{O}(m)$ time and must end in one of the following three cases:
        \begin{enumerate}
            \item We certify that $G$ has $\mu$-expansion $\Phi^{\mu}(G)=\Omega(\phi)$. 
            \item We find a cut $(A,\comp{A})$ in $G$ of $\mu$-expansion $\Phi^{\mu}_G(A,\comp{A})=O(\phi\log n)$, and $\mu(A), \mu(\comp{A})$ are both $\Omega\!\left(\frac{\mu(V)}{\log n}\right)$, i.e, we find a relatively balanced low $\mu$-expansion cut.
            \item We find a cut $(A,\comp{A})$ with $0<\mu(\bar A)\le \frac{\mu(V)}{2}$, $\Phi^{\mu}_G(A,\comp{A})=O(\phi\log n)$, and $G[A]$ is $\Omega(\phi)$-expander with respect to $\mu$. 
        \end{enumerate}
        The correctness of the algorithm holds with high probability.
    \end{theorem}
    
    \begin{proof}[Proof of Theorem~\ref{theorem:main-result3}]
        We apply Theorem~\ref{theorem:cut_matching}. Cases (1) and (2) translate directly to Cases (1) and~(2) of Theorem~\ref{theorem:main-result3}. If Theorem~\ref{theorem:cut_matching} terminates with Case (3), we use Theorem~\ref{theorem:trimming}. We have $\mu(\comp{A})\le\frac{\mu(V)}{10c_0\log n}$, and $\Phi_G^\mu(A, \comp{A})\le c_0\phi\log n$. This means that for large enough $n$,  $|E(A, \comp{A})|\le\frac{\mu(V)\cdot\phi}{10} \le \frac{\mu(A)\cdot\phi}{9}$.
        The trimming step takes $\tilde{O}\left(m\right)$ time, and we return the cut $(A', \comp{A'})$ as Case~(3). Indeed, $G[A']$ is a $\Omega(\phi)$ expander with respect to $\mu$. Additionally, for large enough $n$, $\mu(\comp{A'}) = \mu(\comp{A}) + \mu(A) - \mu(A')\le\mu(\comp{A}) + \frac{4|E(A, \comp{A})|}{\phi} \le \frac{\mu(V)}{10c_0\log n}+\frac{2\mu(V)}{5}\le\frac{\mu(V)}{2}$. 
        This means that $\comp{A'}$ is still the smaller side of the cut, with respect to $\mu$. By Theorem~\ref{theorem:trimming}, we have that $|E(A', \comp{A'})|\le 2|E(A, \comp{A})|$, so $\Phi_G^\mu(A', \comp{A'}) = \frac{|E(A', \comp{A'})|}{\mu(\comp{A'})} \le 2\frac{|E(A, \comp{A})|}{\mu(\comp{A})} = 2\Phi_G^\mu(A, \comp{A}) = O(\phi\log n)$. 
    \end{proof}

    Note that case $(3)$ may return a balanced cut. We get the following stronger version of Theorem~\ref{theorem:main-result3} by classifying case $(3)$ as case $(2)$ whenever $\muG{}{\comp{A}}>\frac{\muG{}{V}}{\log n}$. 

    \begin{corollary}
    \label{theorem:main-result4}
        Given a graph $G=(V,E)$ with $n$ vertices and $m$ edges, a parameter $\phi > 0$, and a vertex measure $\mu : V \to \RR_{\ge 0}$ satisfying Assumption \ref{assumption:mu}, there exists a randomized algorithm which takes $\tilde{O}(m)$ time and must end in one of the following three cases:
        \begin{enumerate}
            \item We certify that $G$ has $\mu$-expansion $\Phi^{\mu}(G)=\Omega(\phi)$. 
            \item We find a cut $(A,\comp{A})$ in $G$ of $\mu$-expansion $\Phi^{\mu}_G(A,\comp{A})=O(\phi\log n)$, and $\mu(A), \mu(\comp{A})$ are both $\Omega\!\left(\frac{\mu(V)}{\log n}\right)$, i.e, we find a relatively balanced low $\mu$-expansion cut.
            \item We find a cut $(A,\comp{A})$ with $0<\mu(\bar A)\le \frac{\mu(V)}{\log n}$, $\Phi^{\mu}_G(A,\comp{A})=O(\phi\log n)$, and $G[A]$ is $\Omega(\phi)$-expander with respect to $\mu$. 
        \end{enumerate}
        The correctness of the algorithm holds with high probability.
    \end{corollary}

    \begin{proof}
        In case Theorem~\ref{theorem:main-result3} terminated with Case (3) and the cut was relatively balanced (specifically, $\frac{\mu(V)}{\log n}<\mu(\comp{A})\le\frac{\mu(V)}{2}$), we instead terminate with Case (2). All conditions of Case (2) are satisfied.
    \end{proof}

    We now show how to get an expander decomposition from Theorem~\ref{theorem:main-result3}. The following is based on the expander decomposition procedure of~\cite{saranurak2019expander, agassy2022expander}.

    \begin{algorithm}[hbt!]
        \caption{Expander Decomposition~\cite{saranurak2019expander, agassy2022expander}}
        \label{algo:expander_decomposition}
        \begin{algorithmic}[1]
            \Function{Decomp}{$G$, $\phi, \muG$}
                \State Apply Theorem~\ref{theorem:main-result4} to $(G,\phi,\muG)$. 
                \If {Case (1): we certify that $\Phi^{\muG}_G\ge \phi$}  
                    \Return $\{V\}$.
                \ElsIf {Case (2): we find a relatively balanced cut $(A,\comp{A})$}
                    \State \Return Decomp($G[A]$, $\phi$, $\muG$) $\cup$ Decomp($G[\comp{A}]$, $\phi$, $\muG$).
                \Else ~{Case (3): we find a very unbalanced cut $(A, \comp{A})$} \Comment{$G[A]$ is an expander.}
                    \State \Return Decomp($G[\comp{A}]$, $\phi$, $\muG$) $\cup \{A\}$.
                \EndIf
            \EndFunction
        \end{algorithmic}
    \end{algorithm}
	
    \begin{proof}[Proof of Theorem~\ref{theorem:expander_decomposition}]
        First, note that by definition, the leaves of the recursion tree give an expander decomposition with $\muG$-expansion of $\Omega(\phi)$.
	    
	    To bound the running time, note that if we get Case (2) of Corollary~\ref{theorem:main-result4} then both sides of the cut have total $\muG{}{\cdot}$ measure of at most $\left(1-\Omega\!\left(\frac{1}{\log n}\right)\!\right)\!\cdot\! \muG{}{V}$. If we get Case (3), then $\muG{}{A} = \Omega(\muG{}{V})$, so the total $\muG{}{\cdot}$ is reduced by a constant factor (\ie, $\muG{}{\comp{A}}\le \left(1-\Omega(1)\right)\cdot \muG{}{V}$). In any case, the total $\muG{}{\cdot}$ measure of the largest component drops by a factor of at least $1-\Omega\!\left(\frac{1}{\log n}\right)$ across each level of the recursion, so (as $\muG{}{V} = O(\poly(n))$, by Assumption~\ref{assumption:mu}) 
        the recursion depth is $O(\log^2 n)$.

	    Since the components on one level of the recursion are all disjoint, the total running time on all components of one level of the recursion is $\tilde{O}(m)$. Since the depth of the recursion is $O(\log^2 n)$ we get that the running time is $\tilde{O}(m)$.

	    To bound the number of edges between expander clusters, observe that in both Case (2) and Case (3), we always cut a component along a cut of $\muG$-expansion $O(\phi \log n)$. Thus, we can charge the edges on the cut to the $\muG{}{\cdot}$ of the vertices on the smaller side of the cut (with respect to $\muG$), so each vertex $v$ is charged $O(\phi \muG{}{v} \log n)$. A vertex can be on the smaller side of a cut at most $O(\log n)$ times, 
        so we can charge each vertex $v$ at most
        $O(\phi \muG{}{v}\log^2 n)$ to pay for all the edges between the final clusters. This bounds the total number of edges between the expanders to be at most $O(\phi \muG{}{V}\log^2 n)$.
	\end{proof}

    \paragraph{AI statement.} All research ideas are original to the authors. GPT-6 Astra identified and corrected a technical bug in the proof of Lemma~\ref{lemma:RST}: source-edge capacities were defined using subtraction, which could produce arbitrarily small positive capacities and invalidate the polynomial capacity-ratio bound needed for the near-linear running time of the fair-cut routine, Theorem~\ref{theorem:fair-cuts}. The correction assigns source weights $m_t(i)=\mu(i)/4$ instead.

\newpage

{\small
	
	\bibliographystyle{alpha}
	\bibliography{refs}
}	
	\newpage
\appendix

\section{Omitted proofs}\label{appendix-proofs}
This appendix contains the proofs that were deferred from the main text. Section~\ref{appendix:omitted-proofs-1} gives algebraic and embedding lemmas that support the analysis of the cut-matching game.
In Section~\ref{appendix:omitted-proofs-2} we introduce the fair-cut framework of~\cite{LNPSsoda13} and use it to analyze flow-based routines: first the flow subroutine in the cut-matching game (Lemma~\ref{lemma:fair-cuts}), and then the trimming step (Theorem~\ref{theorem:trimming}).

\subsection{Algebraic and Embedding Proofs}\label{appendix:omitted-proofs-1}

    \begin{proof}[Proof of Lemma~\ref{lemma:RST}]
        Let $L = \{ i\in A_t \mid u_i < 0 \}$ and $R = \{ i\in A_t \mid u_i \ge 0 \}$ and assume w.l.o.g. that $\muG{}{L} \le \muG{}{R}$. In particular, $\muG{}{L} \le \muG{}{A_t}/2 \le \muG{}{R}$. For a set $B \subseteq A_t$, denote $P_B = \sum_{i\in B} \dG{}{i}u_i^2$. We divide into cases.

        \textbf{Case $P_L \ge \frac{1}{20}P_{A_t}$:} We set $\eta = 0$, $A_t^r=R$, and $\bar{m}_t(i) = \muG{}{i}$ for $i\in A_t^r$. We set $A_t^l=L$ and $m_t(i)=\muG{}{i}/4$ for $i\in A_t^l$.

        It is easy to verify that Properties $(1)$ and $(2)$ of the lemma hold. Property $(3)$ follows from $\muG{}{R}\ge\muG{}{A_t}/2$ and $m_t(A_t^l)=\muG{}{L}/4\le\muG{}{A_t}/8$. Property $(4)$ holds because of the choice $\eta=0$. Finally, Property $(5)$ holds since $\sum_{i\in A_t^l}m_t(i)u_i^2=P_L/4\ge P_{A_t}/80$.

        \textbf{Case $P_L < \frac{1}{20}P_{A_t}$:} Therefore $P_R > \frac{19}{20}P_{A_t}$. By the statement of the lemma, we have that $\sum_{i\in A_t}{\dG{}{i} u_i} = 0$. Denote by $\Delta$ the total distance from zero, that is $\Delta = \sum_{i \in A_t} \dG{}{i}|u_i|$. Observe that $\sum_{i \in L} \dG{}{i}|u_i| = \sum_{i \in R} \dG{}{i}|u_i| = \Delta/2$. Let $M = \muG{}{A_t}$. We set:
        \begin{itemize}
            \item $\eta = 4\Delta/M$.
            \item $A_t^r = \{i \in A_t \mid  u_i \le \eta \}$, and $\bar{m}_t(i) = \muG{}{i}$ for $i\in A_t^r$. 
            \item $A_t^l = \{i \in A_t \mid u_i \ge 6\Delta/M\}$, and $m_t(i) = \muG{}{i}$ for $i\in A_t^l$.
        \end{itemize}
        
        It is easy to see that Properties $(1),(2)$ of the lemma hold. Since $\Delta>0$ in this case and
        \[
            \muG{}{A_t^l}\frac{6\Delta}{M}
            \le \sum_{i\in A_t^l}\muG{}{i}u_i
            \le \frac{\Delta}{2},
        \]
        we have $m_t(A_t^l)=\muG{}{A_t^l}\le M/12\le M/8$, proving the first bound in Property~$(3)$. Note that $\muG{}{R\setminus A_t^r} \le \muG{}{R}/2$, as otherwise
        \begin{align*}
            \Delta = \sum_{i \in A_t} \dG{}{i}|u_i| 
            \ge
            \sum_{i \in R \setminus A_t^r} \dG{}{i}|u_i|
            \ge
            \sum_{i \in R \setminus A_t^r} \dG{}{i}\frac{4\Delta}{M}
            =
            \muG{}{R \setminus A_t^r}\frac{4\Delta}{M}
            >
            \frac{2\muG{}{R}\Delta}{M}
            \ge
            \Delta,
        \end{align*}
        a contradiction. 
        Therefore, $\muG{}{A_t^r} = \muG{}{L} + \muG{}{R\cap A_t^r} \ge (\muG{}{L} + \muG{}{R})/2 = \muG{}{A_t}/2$, proving property $(3)$ of the lemma. To see property $(4)$, let $i\in A_t^l$. Observe that $u_i \ge  6\Delta/M$ and $\eta =  4\Delta/M$, so $u_i - \eta \ge \frac{1}{3}u_i$ and therefore $(u_i-\eta)^2 \ge \frac{1}{9}u_i^2$. We now prove property $(5)$. We show that $P_{A_t^l}\ge \frac{1}{20}P_{A_t}$, which directly implies Property~$(5)$ because $m_t(i)=\muG{}{i}$ for every $i\in A_t^l$. By applying the Cauchy–Schwarz inequality on $a_i = \sqrt{\muG{}{i}}u_i$ and $b_i = \sqrt{\muG{}{i}}$ for $i\in L$, we get
        \begin{align*}
            P_{L} =\sum_{i\in L} \muG{}{i}u_i^2 
            =
            \langle\Vec{a},\Vec{a}\rangle \ge \frac{\langle \Vec{a}, \Vec{b} \rangle^2}{\langle\Vec{b},\Vec{b}\rangle}
            =
            \frac{\Delta^2}{4\muG{}{L}}\ge \frac{\Delta^2}{4M},
        \end{align*}

        and therefore
        
        \begin{align*}
            P_{R\setminus A_t^l} =\sum_{i\in R\setminus A_t^l} \muG{}{i}u_i^2 \le \frac{6\Delta}{M}\sum_{i \in R\setminus A_t^l} \muG{}{i} u_i
            \le 
            \frac{6\Delta}{M}\sum_{i \in R} \muG{}{i} u_i 
            =
            \frac{3\Delta^2}{M}
            \le 12 P_L.
        \end{align*}
        Thus, $P_{A_t^l} = P_R - P_{R\setminus A_t^l} > \frac{19}{20}P_{A_t}-12P_L \ge \frac{7}{20}P_{A_t}\ge \frac{1}{20}P_{A_t}$.
    \end{proof}

    \begin{proof}[Proof of Lemma~\ref{lemma:F_t_embeddable_in_G_t_step1}]
    $ $\newline
        \begin{enumerate}
            \item Recall that $\N{t} = \frac{\delta - 1}{\delta}U + \frac{1}{\delta}\M{t}$. Let $P:E\times V \to \RR_{\ge 0}$ be a routing of $\F$ in $H$ (where $P((u,v), w)$ indicates how much of $w$'s commodity goes through the edge $(u,v)\in E$ from $u$ to $v$). For brevity, we denote $ w_{uv} \coloneqq \MatrixSimple{M_t}{u,v}$, for every $u,v\in V$.
            Note that $\N{t} U^{-1} F$ is the flow matrix obtained by performing a weighted average on the rows of $F$ described by $\M{t}$. Explicitly, for every $v\in V, a \in V$, we have 
            \[(\N{t}U^{-1} F)(v,a) = \Matrix{\frac{\delta-1}{\delta} F + \frac{1}{\delta} \M{t}U^{-1} F}{v,a} = \frac{\delta-1}{\delta}F(v,a) + \frac{1}{\delta}\sum_{u : \{u,v\}\in \M{t}}{\frac{w_{vu}}{\muG{}{u}}F(u,a)}.\footnote{Recall the notation $\M{t}{u,v} = \M{t}{v,u} = w_{uv} = w_{vu}$, for every $u\neq v$.}\]
            The precise construction of the new embedding 
           of $\N{t}\cdot U^{-1} F =\frac{\delta - 1}{\delta}F + \frac{1}{\delta}\M{t} U^{-1}F$ is as follows. 
            \begin{algorithm}[H]
            \begin{algorithmic}[]
                \State $P' \leftarrow \frac{\delta-1}{\delta}\cdot P$.\For{$\{a,b\}\in \M{t}$}
                    \State $P'((a,b),a) \leftarrow P'((a, b), a) + \frac{w_{ab}}{\delta}$.
                    \State $P'((b,a),b) \leftarrow P'((b, a), b) + \frac{w_{ab}}{\delta}$.
                \EndFor
                \For{$\{a,b\}\in \M{t}$}
                    \For{$(u,v)\in H$}
                        \State $P'((u,v),a) \leftarrow P'((u,v),a) + \frac{w_{ab}}{\delta\dG{}{b}}P((u,v),b)$.
                        \State $P'((u,v),b) \leftarrow P'((u,v),b) + \frac{w_{ab}}{\delta\dG{}{a}}P((u,v),a)$.
                    \EndFor
                \EndFor
            \end{algorithmic}
            \end{algorithm}
    
            The argument that we indeed obtain an embedding of $N_t U^{-1} F=\frac{\delta-1}{\delta}F+\frac{1}{\delta} M_t U^{-1} F$ is as follows.
            We think of $P'$ in  stages. In the first stage, we scale $P$ by $\frac{\delta-1}{\delta}$. This routes $\frac{\delta-1}{\delta}\MatrixSimple{F}{v,a}$ units of flow (of $v$'s commodity) from $v$ to $a$ for every $v$ and $a$. After this stage, each vertex $v\in V$ currently sends $\frac{(\delta-1)}{\delta}\dG{}{v}$ units of its commodity. In the next stage we wish to route an additional $\frac{1}{\delta}\sum_{u : \{u,v\}\in \M{t}}{\frac{w_{vu}}{\muG{}{u}}\MatrixSimple{F}{u,a}}$ units from $v$ to $a$. To this end, we first move $\frac{w_{vu}}{\delta}$ units from $v$'s commodity to each $u$ with $\{u,v\}\in \M{t}$ (this flow is sent through the edge $\{v,u\}\in H\cup \M{t}$ from $v$ to $u$). Now, each vertex $u\in V$ ``mixes'' the commodities it got from its neighbors and routes the $\frac{\muG{}{u}}{\delta}$ ``new'' units according to $P$, as if they were of its own commodity. 
            Thus, $\frac{\MatrixSimple{F}{u,a}}{\delta}$ is the total flow sent from $u$ to $a$ when it routes the $\frac{\muG{}{u}}{\delta}$ ``new'' units and the ``share'' of $v$'s commodity from this flow is $\frac{w_{vu}}{\muG{}{u}}$.
            It follows that out of the $\frac{w_{vu}}{\delta}$ units $v$ sent to $u$, exactly $\frac{w_{vu}}{\muG{}{u}}\frac{\MatrixSimple{F}{u,a}}{\delta}$ units go to $a$.

            As for the congestion, note that on each edge of the matching $\M{t}$, we route $\frac{w_{ab}}{\delta}$ units of flow in each direction, so the congestion on each such edge is $\frac{2}{\delta}$.
            For each (directed) arc $(u, v)\in H$ and $a\in V$, we have $P'((u,v),a) = \frac{\delta - 1}{\delta}\cdot P((u,v),a) + \sum_{b : \{a,b\}\in \M{t}}{\frac{w_{ab}}{\delta\dG{}{b}}P((u,v),b)}$.
            Therefore, on each arc $(u,v)\in H$, the congestion is 
            \begin{align*}
                \sum_{a\in V}{P'((u,v),a)} &= \frac{\delta - 1}{\delta}\cdot \sum_{a\in V}{P((u,v),a)} + \sum_{a\in V}{\sum_{b : \{a,b\}\in \M{t}}{\frac{w_{ab}}{\delta\dG{}{b}}P((u,v),b)}}
                \\
                &= \frac{\delta - 1}{\delta}\cdot \sum_{a\in V}{P((u,v),a)} + \frac{1}{\delta}\sum_{b\in V}{P((u,v),b)} 
                \\
                &= \frac{\delta - 1}{\delta}\cdot c(u,v) + \frac{1}{\delta}c(u,v) = c(u,v) \ .
            \end{align*}
            \item  Let $P$ be the routing of $F$ in $H$. Note that $F U^{-1}\N{t}=\frac{\delta - 1}{\delta}F + \frac{1}{\delta}FU^{-1}\M{t}$ is the flow matrix obtained by performing a weighted average on the columns of $F$ matched by $U^{-1} \M{t}$, i.e. we average the flow received by matched vertices.
            
            We define a routing $P'$ of $F U^{-1}\N{t}$ in $H\cup \M{t}$ as follows:
            \begin{algorithm}[H]
            \begin{algorithmic}[]
                \State $P' \leftarrow P$.
                \For{$\{a,b\}\in \M{t}$}
                    \For{$u\in V$}
                        \State $P'((a,b),u) \leftarrow \frac{w_{ab}}{\delta \muG{}{a}}\cdot \MatrixSimple{F}{u,a}$.
                        \State $P'((b,a),u) \leftarrow \frac{w_{ba}}{\delta \muG{}{b}}\cdot \MatrixSimple{F}{u,b}$.
                    \EndFor
                \EndFor
            \end{algorithmic}
            \end{algorithm}
            That is, $P'$ routes the same as $P$ over edges in $H$.
            For every $\{a,b\}\in \M{t}, u \in V$, we set $P'((a,b),u) = \frac{w_{ab}}{\delta \muG{}{a}}\cdot \MatrixSimple{F}{u,a}$, and symmetrically, $P'((b,a),u) = \frac{w_{ba}}{\delta \muG{}{b}}\cdot \MatrixSimple{F}{u,b}$.
            That is, for every $\{a,b\}\in \M{t}$, $a$ (respectively, $b$) sends $\frac{w_{ab}}{\delta}$ units of its received flow, which is a mix of commodities (commodity $u$ has a $\frac{F(u,a)}{\muG{}{a}}$ share in this mix), to $b$ (respectively, $a$). 
            Note that $\Matrix{F U^{-1} \N{t}}{u, a} = \frac{\delta - 1}{\delta}\MatrixSimple{F}{u, a} + \frac{1}{\delta}\sum_{b : \{a,b\}\in \M{t}}{\frac{w_{ab}}{\muG{}{b}}\MatrixSimple{F}{u, b}}$. Thus, $P'$ routes $F U^{-1}\N{t}$ in $H$.
            
            Note that the congestion of $P'$ on each edge $e\in H$ is still $c(e)$, and on each arc $(a,b)\in \M{t}$, we have $\sum_{u\in V}{P'((a,b),u)} = \sum_{u\in V}{\frac{w_{ab}}{\delta \muG{}{a}}\cdot \MatrixSimple{F}{u,a}} = \frac{w_{ab}}{\delta}$ (the last equality follows as $F$ is $\muG$-stochastic). Therefore at most $\frac{w_{ab}}{\delta}$ flow was routed in each direction, so the congestion on the edge $\{u, v\}\in \M{t}$ is at most $\frac{2}{\delta}$.
            
            \item Note that if $F$ and $\N{t}$ are $\dG$-stochastic (see part $2$ of Lemma~\ref{lemma:basic_properties}) then $N_{t} \cdot U^{-1} F$ is also $\dG$-stochastic, since:  
            \begin{align*}
                N_{t} U^{-1} F \cdot \1_n &= N_{t} U^{-1} \dG = \dG
                \\
                \1_n' \cdot N_{t} U^{-1} F &= \dG' \cdot U^{-1} F =  \dG' \ .
            \end{align*}
            Therefore we can apply Parts (1) and (2) to get the result.
        \end{enumerate}
    \end{proof}

    \begin{proof}[Proof of Lemma~\ref{lemma:X_as_normalized_laplacian}]
         Let $\X{t} = \Uminushalf N_t \Uminushalf$. Therefore 
            \[
            \X{t} = \normalized{\left(\frac{\delta-1}{\delta}\U + \frac{1}{\delta}\M{t}\right)} = 
            \frac{\delta-1}{\delta}I_{\ter} + \frac{1}{\delta}\normalized{\M{t}} = 
            I_{\ter} - \frac{1}{\delta}(I_{\ter} - \normalized{\M{t}}) = 
            I_{\ter} - \frac{1}{\delta}\mathcal{N}(\M{t}).
            \]
            Observe that $\mathcal{N}(\M{t})$ and $I_{\ter}-\X{t}^{4\delta}$ have the same eigenvectors. Indeed, since both matrices are $\ter$-blocked, their null space contains all vectors whose non-zero coordinates are outside $\ter$. Consider an eigenvector $v\in \RR^n$, whose non-zero coordinates are in $\ter$, of $\mathcal{N}(\M{t})$, with an eigenvalue of $\lambda$. Then,
            \begin{align*}
                \X{t} v = \left(I_{\ter}-\frac{1}{\delta}\mathcal{N}(\M{t})\right)v = v-\frac{\lambda}{\delta}v = \left(1-\frac{\lambda}{\delta}\right)v \ .
            \end{align*}
            Therefore, $(I_{\ter}-\X{t}^{4\delta}) v = \left(1-\left(1-\frac{\lambda}{\delta}\right)^{4\delta}\right)v$.
            We can see that an eigenvector $v\in \RR^n$ of $\mathcal{N}(\M{t})$, with eigenvalue $\lambda$, is an eigenvector of $I_{\ter}-\X{t}^{4\delta}$ with eigenvalue $1-(1-\frac{\lambda}{\delta})^{4\delta}$. A known property of normalized Laplacians (see Lemma~\ref{lemma:normalized-laplacian-eigenvalues}) 
            is that all of their eigenvalues are in the interval $[0,2]$, so $\lambda \in [0,2]$. For any such $\lambda$ it holds that $1-(1-\frac{\lambda}{\delta})^{4\delta}\ge1-\frac{1}{e^{4\lambda}}$. 
            Indeed, for $\lambda=0$ an equality is achieved, and for $\lambda\in (0,2]$, note that $\delta = \Theta(\log n)$, 
            so $\frac{\delta}{\lambda} \ge 1$ and we get
            \[
            1-\left(1-\frac{\lambda}{\delta}\right)^{4\delta} = 1-\left(\left(1-\frac{\lambda}{\delta}\right)^\frac{\delta}{\lambda}\right)^{4\lambda} \ge
            1- \frac{1}{e^{4\lambda}}
            \]
            Simple calculus shows that $1-\frac{1}{e^{4x}} \ge \frac{1}{3}x$ for any $x\in [0,2]$. 
            Therefore, we get that any eigenvector $v\in \RR^n$ of $\mathcal{N}(\M{t})$, with eigenvalue $\lambda$, is also an eigenvector of $I-\X{t}^{4\delta}$ with eigenvalue at least $\frac{1}{3}\lambda$. The result follows since both matrices are symmetric and therefore have a spanning basis of eigenvectors.
    \end{proof}

\subsection{Fair Cuts}\label{appendix:omitted-proofs-2}

    In this section, we prove Lemma~\ref{lemma:fair-cuts} and Theorem~\ref{theorem:trimming}, concerning the flow routines in the cut-matching step and the trimming step. 
    The proofs are based on fair cuts, a stronger notion of approximate min-cut, introduced by~\cite{LNPSsoda13}.

    \begin{definition}[Fair Cut, \cite{LNPSsoda13,li2025simple}]
        Let $G=(V, E)$ be an undirected graph with edge capacities $c\in \RR^E_{>0}$. Let $s, t$ be two vertices in $V$. For any parameter $\alpha \ge 1$, we say that a cut $(S, T)$ is a \emph{$\alpha$-fair $(s, t)$-cut} if there exists a feasible $(s, t)$-flow $f$ such that $f(u, v)\ge \frac{1}{\alpha}\cdot c(u, v)$ for every $(u, v)\in E(S, T)$, where $u\in S$ and $v\in T$.
    \end{definition}

    \begin{theorem}[Fair Cut, \texorpdfstring{\cite[Theorem 1.1]{LNPSsoda13}}{[LNPS23, Theorem 1.1]}]
    \label{theorem:fair-cuts}
        Given a graph $G=(V, E)$ whose positive edge capacities have ratio $\poly(n)$, two vertices $s, t\in V$ and $\epsilon\in (0, 1]$, we can compute with high probability a $(1+\epsilon)$-fair $(s, t)$-cut, together with a corresponding feasible flow $f$, in $\tilde{O}(\frac{m}{\epsilon^3})$ time.
    \end{theorem}

    \begin{proof}[Proof of Lemma~\ref{lemma:fair-cuts}]
        Let $\alpha=0.1$. Consider the following $(s, t)$-flow problem on a new graph $H=(V_H = V\cup \{ s,t\}, E_H,w_H)$ constructed as follows. We begin with $G$ and set its edge capacities to $\frac{c}{1+\alpha}$. We then add the source $s$ and connect it to every $v\in A^{l}$ with capacity $m(v)$. Finally, we add the sink $t$ and connect it to every $v\in A^r$ with capacity $\frac{1}{1+\alpha}\bar{m}(v)$. 

        Using Theorem \ref{theorem:fair-cuts}, we compute (with high probability) a $(1+\alpha)$-fair cut $(S',T')$ with a corresponding feasible flow $f$ in $\tilde{O}(|E|)$ time. We split into the following two cases:

        \textbf{Case 1: $\muG{}{S'\setminus\{s\}}=0.$}
            Every $v\in A^l$ with $m(v)>0$ lies in $T'$, since $m(v)\le\muG{}{v}$. By the definition of fair cut, for every $v\in A^l$, $f$ sends from $s$ to $v$ at least $\frac{1}{1+\alpha}m(v)$ units of flow. By decomposing $f$ into flow paths and reducing the flow sent on each path, we obtain a flow $f'$ for $H$ that sends from $s$ to $v$ exactly $\frac{1}{1+\alpha}m(v)$ units of flow, for every $v\in A^l$. We claim that $f'' = (1+\alpha)f'$ is a feasible flow for $\Pi(G)$. Indeed, in $f''$, every $v\in A^l$ receives $m(v)$ units of flow from $s$, every $v\in A^r$ sends at most $(1+\alpha)\cdot \frac{1}{1+\alpha}\bar{m}(v) = \bar{m}(v)$ units of flow to $t$, and every $e\in E$ carries at most $(1+\alpha)\cdot \frac{1}{1+\alpha}c=c$ units of flow. The algorithm returns $f''$ as Case~$(1)$ of the lemma.
        
        \smallskip
        \textbf{Case 2: $\muG{}{S'\setminus\{s\}}>0.$}
            We split the cut edges $E_H(S',T')$ into three sets: (i) the edges of $G$ crossing the cut: $E_m = E_G(S'\setminus \{s\}, T'\setminus \{t\})$, (ii) the cut edges incident to $s$: $E_s = \{(s,v)\in E_H\mid v\in T'\}$, and (iii) the cut edges incident to $t$: $E_t = \{(v,t)\in E_H\mid v\in S'\}$. Let $\bar{E}_s$ be the incident edges to $s$ that are not in $E_s$. 

            We decompose $f$ into paths. By the definition of a fair cut, every edge $(s,v)\in E_s$ sends at least $\frac{1}{1+\alpha}m(v)$ units of flow. We scale down the flow in the flow paths as follows. We delete all flow paths that start with an edge of $\bar{E}_s$. 
            The rest of the flow paths are scaled down until each edge $(s,v)\in E_s$ sends exactly  $\frac{1}{1+\alpha}m(v)$ units of flow. This results in a flow $f'$. The algorithm then returns $f'' = (1+\alpha)f'$ and the cut $S = S'\setminus \{s\}$.

            Similarly to the previous case, $f''$ is a feasible flow for the flow problem $\Pi(G\setminus S)$. It remains to prove that $\Phi_{G}^{\mu}(S, V\setminus S) \defeq \frac{|E_m|}{\min \{\muG{}{S}, \muG{}{V\setminus S} \}}\le\frac{7}{c}$. We begin by upper bounding the numerator. For the remainder of the analysis, we consider again the flow $f$ and its path decomposition. By the definition of a fair cut, every flow path that traverses through $E_m$ must start with an edge in $\bar{E}_s$, and every edge in $E_m$ is at least $\frac{1}{1+\alpha}$-saturated. Therefore, 
            \begin{align}
            \label{eq:matching-1}
                \frac{c}{1+\alpha}\cdot |E_G(S, V\setminus S)| = w_H(E_m) \le (1+\alpha) w_H(\bar{E}_s) = (1+\alpha) m_t(S) \le (1+\alpha) \muG{}{S}.
            \end{align} 
            To bound the denominator we prove that $\muG{}{V\setminus S}= \Omega(\muG{}{S})$. Recall the assumption of the lemma $\bar{m}(A^r)\ge \muG{}{V}/2$. We prove that $\bar{m}(A^r \cap S)$ is small. Indeed, for every $v \in A^r \cap S$, $(v,t)\in E_t$ is crossing the cut $(S',T')$ and therefore is $\frac{1}{1+\alpha}$-saturated. Since the total flow crossing $(S',T')$ is bounded by the total source capacity (which is $m(A^l)$), we get that 
            \begin{align*}
                \frac{1}{1+\alpha}\bar{m}(A^r\cap S) = w_H(E_t) \le (1+\alpha)m(A^l) \le \frac{1+\alpha}{8} \muG{}{V},
            \end{align*}
            where the last inequality follows from the assumption of the lemma. Putting everything together,

            \begin{align*}
                \muG{}{V\setminus S} \ge \bar{m}(V\setminus S) =
                \bar{m}(A^r) - \bar{m}(A^r \cap S) \ge \frac{\muG{}{V}}{2} - \frac{(1+\alpha)^2}{8}\muG{}{V} \ge \frac{1}{3}\muG{}{V} \ge
                \frac{1}{3}\muG{}{S}.
            \end{align*}
            Together with~(\ref{eq:matching-1}), we conclude that

            \begin{align*}
                \Phi_{G}^{\mu}(S, V\setminus S) \defeq 
                \frac{|E_G(S, V\setminus S)|}{\min \{\muG{}{S}, \muG{}{V\setminus S} \}} \le
                \frac{\frac{(1+\alpha)^2}{c}\muG{}{S}}{\frac{1}{3}\muG{}{S}}
                \le \frac{7}{c}.
            \end{align*}
    \end{proof}

    As a direct corollary, we observe that if Lemma~\ref{lemma:fair-cuts} terminates with Case (2), we can ``break" the flow paths at the (fair-)cut $(S,V\setminus S)$ and obtain two flows satisfying the following guarantees.

    \begin{corollary}
    \label{cor:flow-from-cut-cut-matching}
        Let $(S,V\setminus S)$ be the cut returned by Lemma~\ref{lemma:fair-cuts}, and assume $S \neq \emptyset$. Then:
        \begin{itemize}
            \item There exists a flow of congestion $O(1)$ in $G[S]$ such that every $v\in S$ sends $ \left( \deg_G-\deg_{G[S]} \right)\!(v)$ units of flow and receives at most $O(\phi\log n\cdot\mu(v))$ units of flow. 
            \item There exists a flow of congestion $O(1)$ in $G[V\setminus S]$ such that every $v\in V\setminus S$ sends $ \left( \deg_G-\deg_{G[V\setminus S]} \right)\!(v)$ units of flow and receives at most $O(\phi\log n\cdot\mu(v))$ units of flow.
        \end{itemize}
    \end{corollary}

    \begin{proof}
        In the notation of the proof of Lemma~\ref{lemma:fair-cuts}, the assumption $S\neq \emptyset$ corresponds to Case (2) of the proof. Let $f$ be the flow returned by the fair-cut routine, and recall  $E_m = E_G(S'\setminus \{s\}, T'\setminus \{t\}) = E_G(S, V\setminus S),$ the set of edges of $G$ that cross the fair-cut. Recall that $f$ is routed in $H$, the auxiliary graph, with constant congestion. In particular, $f$ may be routed in $G$ with congestion $O(c)$.

        To construct the two flows required by the corollary, we proceed as follows. Decompose $f$ into flow paths and remove the paths that do not traverse any edge of $E_m$. Denote the remaining flow by $f_1$. By the definition of a fair-cut, every edge $(u,v)\in E_m$ is at least $\frac{1}{1+\alpha}$-saturated (recall $\alpha = 0.1$). Hence $f_1$ routes at least $\frac{1}{(1+\alpha)^2}c = \Theta \left(\frac{1}{(1+\alpha)^2\phi \log n} \right)$ units of flow through each $(u,v)\in E_m$ in the direction from $S$ to $V\setminus S$. We then scale down the flow paths of $f_1$ so exactly $\frac{1}{(1+\alpha)^2}c$ units of flow are routed on each edge of $E_m$. Denote this scaled flow by $f_2$. Finally, define $f_3 \defeq \frac{(1+\alpha)^2}{c}f_2$. Then $f_3$ routes exactly one unit of flow across each edge in $E_m$ and has congestion $O(1)$ in $G$. Moreover, for every $v\in V$ the source and sink capacities at $v$ in $f_3$ are bounded by $O\left(\frac{(1+\alpha)^2}{c} \muG{}{v} \right) = O(\phi \log n \muG{}{v})$. 

        Observe that $f_3$ consists of flow paths that cross each edge of the cut $(S, V \setminus S)$ once.
        We obtain the two desired flows as follows: 
        \begin{itemize}
            \item The first flow is obtained by taking the prefix of each path until the crossing edge, deleting $s$ and then reversing the path so that the flow is routed from the cut $E_m = E_G(S,V\setminus S)$ to $S$.
            \item The second flow is obtained by taking the suffix of each path that starts at the crossing edge, deleting $t$.
        \end{itemize}
    \end{proof}

    Recall that during the cut-matching game, the matching player produces cuts $S_1,\ldots,S_t$ (up to some $t < T$) by applying Lemma~\ref{lemma:fair-cuts}. Since these cuts are disjoint, we get the following corollary.

    \begin{corollary}
    \label{cor:flow-from-cut-cut-matching-global}
        Let $(A,\comp{A})$ be the cut returned by Case (2) or Case (3) of Theorem~\ref{theorem:main-result}. Then, there exists a flow of congestion $O(1)$ in $G[\comp{A}]$ such that every $v\in \comp{A}$ sends $ \left( \deg_G-\deg_{G[\comp{A}]} \right)\!(v)$ units of flow and receives at most $O(\phi\log n\cdot\mu(v))$ units of flow. 
    \end{corollary}

    \begin{proof}
        Follows from Corollary~\ref{cor:flow-from-cut-cut-matching}.
    \end{proof}

    \begin{proof}[Proof of Theorem~\ref{theorem:trimming}]
        Assume $|E(A,\comp{A})| > 0$, as otherwise $A$ is a $(\phi,\mu)$-expander and we set $A'=A$. 
        Consider the following $(s, t)$-flow problem on a new graph $H=(V_H = A\cup \{ s,t\}, E_H)$ constructed as follows. 
        Start from $G$, and contract $V\setminus A$ into a single vertex, labeled as the source $s$ (ignoring self loops at $s$). 
        Next, set the capacity of each edge ($G[A]$ and edges touching $s$)
        to be $\frac{3}{\phi}$. Finally, add a new sink vertex $t$ and connect it to each vertex $v\in A$ with an edge of capacity $\mu(v)$. Let $\alpha=0.1$, and compute a $(1+\alpha)$-fair cut $(S, T)$ of $H$ 
        in time $\tilde{O}(m)$. Let $A'=T\setminus\{t\}$.\footnote{We later show that $A'\neq \emptyset$.} We now show that $A'$ satisfies the requirements of the lemma.
        
        First, suppose for contradiction that $G[A']$
        is not a $\frac{\phi}{6}$-expander with respect to $\mu$. Then, there is a violating set $U\subseteq A'$ (see Figure~\ref{fig:A-Aprime-U}), satisfying $0 < \mu(U)\le \mu(A'\setminus U)$ and 
        \[
            |E(U, A'\setminus U)|<\frac{\phi}{6}\mu(U) ~.
        \]
        Since $A$ is a near $(\phi,\muG)$-expander with respect to $\mu$ then so is $A'$ 
        and therefore,

     \begin{figure}[t]
        \begin{center}
    \includegraphics[scale=0.65]{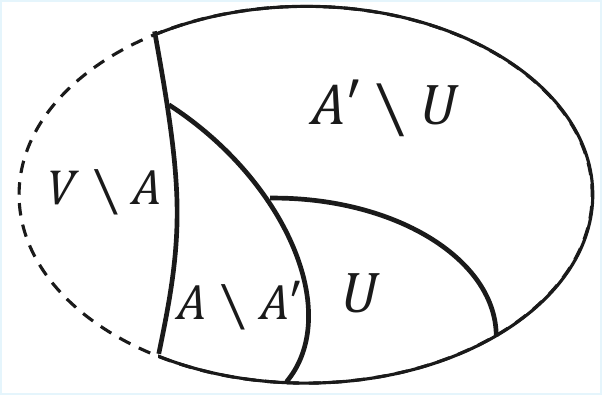}
        \end{center}
        \caption{Visualization of the sets $A,A'$ and $U$.}
        \label{fig:A-Aprime-U}
    \end{figure}
        
        \[
            |E(U, V\setminus U)|\ge\phi\mu(U) ~.
        \]
        Taking the difference of the two inequalities above, 
        \[
            |E(U, V\setminus A')| = |E(U, V\setminus U)| - |E(U, A'\setminus U)| \ge \frac{5\phi}{6}\mu(U) ~. 
        \]
        Since $(S, T)$ is a $(1+\alpha)$-fair cut, there is a feasible flow $f$ that saturates each edge of $E_H(S, T)$ up to a factor of $\frac{1}{1+\alpha}$. Each edge $(u, v)$ in $E(U, V\setminus A')$ corresponds to an edge in $E_H(S, T)$ of capacity $\frac{3}{\phi}$, and the flow $f$ must send at least $\frac{1}{1+\alpha}\cdot\frac{3}{\phi}\ge \frac{2}{\phi}$ flow along the edge (in the direction from $S$ to $T$). In total, the amount of flow entering $U$ in $H$ is at least 
        \[
            \frac{2}{\phi}|E(U, V\setminus A')|\ge\frac{2}{\phi}\cdot\frac{5\phi}{6}\mu(U)=\frac{5}{3}\mu(U) ~.
        \]
        On the other hand, at most $\mu(U)$ flow can leave $U$ along the edges incident to $t$, and at most 
        \[
            \frac{3}{\phi}|E(U, A'\setminus U)|\le\frac{3}{\phi}\cdot\frac{\phi}{6}\mu(U) = \frac{1}{2}\mu(U)
        \]
        flow can cross from $U$ to $A'\setminus U$. This totals at most $\frac{3}{2}\mu(U)$ flow that can exit $U$, which is strictly less than the $\ge\frac{5}{3}\mu(U)$ flow that enters $U$ (as $\mu(U)>0$), a contradiction. Thus $G[A']$ is a $\frac{\phi}{6}$-expander.
        
        Finally, we show the properties $\mu(A')\ge\mu(A)-\frac{4}{\phi}|E(A, \comp{A})|$ and $|E(A', \comp{A'})|\le 2|E(A, \comp{A})|$ promised by the lemma. Since $(S, T)$ is a $(1+\alpha)$-fair cut, it is in particular a $(1+\alpha)$-approximate $(s, t)$-mincut. 
        Since $(\{s\}, V_H\setminus \{s\})$ is an $(s, t)$-cut of capacity $\frac{3}{\phi}|E(A, \comp{A})|$, it follows that the cut $(S, T)$ has capacity at most $(1+\alpha)\frac{3}{\phi}\cdot|E(A, \comp{A})|$. To establish the first inequality of the theorem, observe that every vertex $v\in A\setminus A'$ lies on the $S$-side of the cut $(S, T)$, and therefore contributes $\mu(v)$ to the capacity of $(S, T)$ through the edge $(v, t)$. Summing over all $v\in A\setminus A'$, we obtain
        \[
            \mu(A\setminus A')\le|E(S, T)|\le(1+\alpha)\cdot\frac{3}{\phi}|E(A, \comp{A})|\le\frac{4}{\phi}|E(A, \comp{A})| ~,
        \]
        which proves the first property. Moreover, since $\mu(A) \ge\frac{9|E(A, \comp{A})|\cdot}{\phi}$, it follows that $A'\neq \emptyset$. 
        For the second inequality, note that each edges $(u, v)$ in $E(A', \comp{A'})$ corresponds to an edge in $E(S, T)$ with  capacity $\frac{3}{\phi}$, so summing over all such edges 
        \[
            \frac{3}{\phi}|E(A', \comp{A'})|\le|E(S, T)|\le(1+\alpha)\frac{3}{\phi}|E(A, \comp{A})| ~,
        \]
        meaning $|E(A', \comp{A'})|\le(1+\alpha)|E(A, \comp{A})|$ which proves the second property.
        
    \end{proof}

    From the auxiliary flow problem defined in the preceding proof, we deduce the following corollary.

    \begin{corollary}
        Let $(A', \comp{A'})$ be the cut returned from the trimming step (Theorem~\ref{theorem:trimming}) in the proof of Theorem \ref{theorem:main-result3}. Then, the following flows exist.
        \begin{enumerate}
            \item A flow of congestion $O(1)$ in $G[A']$, such that every $v\in A'$ sends $\left( \deg_G-\deg_{G[A']}\right)\!(v)$ units of flow and receives at most $O(\phi\cdot\muG{}{v})$ units of flow. 
            \item A flow of congestion $O(1)$ in $G[\comp{A'}]$ such that every $v\in \comp{A'}$ sends $ \left( \deg_G-\deg_{G[\comp{A'}]} \right)\!(v)$ units of flow and every $v\in \comp{A}$ receives at most $O(\phi\log n\cdot\mu(v))$ units of flow.\footnote{Note that vertices in $\comp{A'}\setminus \comp{A}$ don't receive any flow.} 
        \end{enumerate} 
    \end{corollary}

    \begin{proof}
         Let $H$ be the auxiliary graph in the proof of Theorem~\ref{theorem:trimming}, and let $f$ denote the flow constructed therein. We decompose $f$ into flow paths from $s$ to $t$ and scale down each flow path so that the flow crossing every edge of $E_H(S,T\setminus \{t\})$ equals $\frac{2}{\phi}$ (because $(S, T)$ is a fair cut). Let $f_1$ denote the resulting flow. Let $f_2 \defeq \frac{\phi}{2}f_1$.  Note that $f_2$ routes exactly one unit of flow across every edge of $E_H(S,T\setminus \{t\})$, at most $\frac{3}{2}$ units of flow across the remaining edges of~$H\setminus \{t\}$, and at most $\frac{\phi }{2}\muG{}{v}$ units across each edge $(v, t)$.
    
        For the first item, we delete from $f_2$ all flow paths that contain an edge $(v,t)$, where $t$ is the sink and $v\in S = \{s\}\cup \comp{A'}\setminus \comp{A}$. Denote the resulting flow by $f_3$.
        The desired flow is then obtained by restricting the flow paths in $f_3$ to their suffix, starting from the edge crossing the fair-cut $(S,T)$.

        We proceed to the second item, and prove it in two steps. 
        We first route flow from the new cut edges $E_G(A',\comp{A'})\setminus E_G(A, \comp{A})$ to $E_G(A, \comp{A})$, and then we route the flow from $E_G(A, \comp{A})$ to $\comp{A}$ so that every $v\in \comp{A}$ receives at most $O(\phi \log n\cdot\muG{}{v})$ units of flow. 
        
        For the first step,  
        observe that there is a flow $f^{(1)}$ of congestion $\frac{3}{2}=O(1)$ in $G[\comp{A'}]$ such that every $v\in \comp{A'}$ sends $\left(\deg_G-\deg_{G[\comp{A'}]}\right)\!(v)$ units of flow and each $u\in \comp{A}$ receives at most $\frac{3}{2}\left(\deg_G-\deg_{G[\comp{A}]}\right)\!(u)$ units of flow. This flow is obtained by restricting the flow paths of $f_3$ to their prefix from $s$ until the first edge in the fair-cut $(S,T)$, and then reversing them, so that they route flow from the cut $(S, T)$ towards $s$. In the last edge $(v_1, s)$ of each flow path, instead of sending the flow to $s$, we send it to the vertex $v_0$ of $G[\comp{A}]$ corresponding to the edge $(v_0, v_1) \in E(A, \comp{A})$.
        
        In the second step, note that the cut $(A,\comp{A})$ was produced in the cut-matching step. Therefore, by Corollary~\ref{cor:flow-from-cut-cut-matching-global}, there exists a flow $f^{(2)}$ of congestion $O(1)$ in $G[\comp{A}]$ such that every $v\in \comp{A}$ sends $\left(\deg_G-\deg_{G[\comp{A}]}\right)\!(v)$ units of flow and receives at most $O(\phi \log n \cdot \muG{}{v})$ units of flow.

        The final flow is as follows. We first route according to $f^{(1)}$.  Each $v\in \comp{A}$ must then send at most $\frac{3}{2}\!\left(\deg_G-\deg_{G[\comp{A}]}\right)\!(v)$ units of flow, which are the units it receives from $f^{(1)}$. We route this flow using a scaled version of $f^{(2)}$, specifically $\frac{3}{2}f^{(2)}$, which incurs a congestion of $O(1)$.
    \end{proof}

    \section{Algebraic Tools}
	\label{appendix:matrix_inequalities}
    \begin{fact}[\texorpdfstring{\cite[Exercise \uppercase\expandafter{\romannumeral 9\relax}.3.3]{bhatia2013matrix}}{[Bha13, Exercise IX.3.3]}]
    \label{fact:book_trace_identities}
    	Let $X,Y,A \in \RR^{n \times n}, m \in \NN$, then 
    	\begin{enumerate}
    	    \item $\tr(XY) = \tr(YX)$.
    	    \item $\tr(A^{2m}) \le \tr((A \cdot A')^m)$.
    	\end{enumerate}
	\end{fact}
    \begin{proof} 
        \begin{enumerate}
	        \item Follows since $\tr(XY) = \sum_{i=1}^n {\sum_{j=1}^n{\X{}{i,j}\Y{}{j,i}}}
	        = \sum_{j=1}^n {\sum_{i=1}^n{\Y{}{j,i}\X{}{i,j}}}
	        = \tr(YX)$.
	        \item Using Schur decomposition, we decompose $A = V U V^{*}$, where $V \in \CC^{n\times n}$ is a unitary matrix and $U\in \CC^{n\times n}$ is an upper triangular matrix. Observe that
	        \[\tr(A^{2m}) = \tr((V U V^{*})^{2m}) = \tr(V U^{2m} V^{*}) = \tr(V^{*} V U^{2m} ) = \tr(U^{2m}) = \sum_{i=1}^n{\U{}{i,i}^{2m}}. \]
	        
	        Let $S = U U^*$, then $\Smat{}{i,i} = \sum_{j=1}^n{|\U{}{i,j}|^2} \ge |\U{}{i,i}|^2$, for every $i$. Moreover, $S$ is PSD Hermetian matrix so $S$ can be decomposed to $S = R D R^*$, where $R$ is a unitary matrix and $D$ is diagonal with non-negative real entries. Since $R$ is unitary, it holds that $\sum_{j=1}^n|\R{}{i,j}|^2 = 1$, for every $i$. By convexity, we get
	        \begin{align*}
	            \Matrix{S^m}{i,i} &= \Matrix{R D^m R^*}{i,i} = \sum_{j=1}^n{|\R{}{i,j}|^2 \D{}{j,j}^m} \ge
    	        \left( \sum_{j=1}^n{|\R{}{i,j}|^2 \D{}{j,j}} \right)^m \\&=
    	        (\Smat{}{i,i})^m \ge |\U{}{i,i}|^{2m}.
	        \end{align*}
	        
	        Therefore 
	        \begin{align*}
	             \tr((A \cdot A')^m) &= \tr((V U U^* V^*)^m) = \tr(V (U U^*)^m V^*) = \tr((U U^*)^m)\\
	             &= \tr(S^m) = \sum_{i=1}^n \Matrix{S^m}{i,i} \ge \sum_{i=1}^n (\U{}{i,i})^{2m} 
	             = \tr(A^{2m}).
	        \end{align*}
	    \end{enumerate}
	\end{proof}
    \begin{lemma}[In the proof of~\texorpdfstring{\cite[Theorem \uppercase\expandafter{\romannumeral 9\relax}.3.5]{bhatia2013matrix}}{[Bha13, Theorem IX.3.5]}]
    \label{lemma:book_inequality}
    	Let $X,Y \in \RR^{n \times n}$ be symmetric matrices. Then for any positive integer $k$,  $\tr((XY)^{2^k}) \le \tr(X^{2^k}Y^{2^k})$.
	\end{lemma}
	\begin{proof}
	    By Fact~\ref{fact:book_trace_identities}, we get
	    \begin{align*}
            \tr((XY)^{2^k}) \le \tr(((XY)(XY)')^{2^{k-1}}) = 
            \tr((XY^2 X)^{2^{k-1}}) = \tr((X^2 Y^2)^{2^{k-1}}).
	    \end{align*}
	    The result follows by induction.
	\end{proof}
    \begin{theorem} [Symmetric Rearrangement; \texorpdfstring{\cite[Theorem A.2]{orecchia2008partitioning}}{[OSVV08, Theorem A.2]}]
	\label{theorem:symmetric_rearrangement}
	    Let $X, Y\in \RR^{n\times n}$ be symmetric matrices. Then for any positive integer $k$,
	    \begin{align*}
	        \tr\left((XYX)^{2^k}\right) \le \tr\left(X^{2^k}Y^{2^k}X^{2^k}  \right)\ .
	    \end{align*}
	\end{theorem}
	\begin{proof}
	    By Lemma~\ref{lemma:book_inequality} and Fact~\ref{fact:book_trace_identities}
	    \begin{align*}
	        \tr\left((XYX)^{2^k}\right) = \tr((X^2Y)^{2^k}) \le 
	        \tr(X^{2^{k+1}}Y^{2^k}) = \tr(X^{2^k}Y^{2^k}X^{2^k})
	    \end{align*}
	\end{proof}
	
	The following lemma is standard.
	\begin{lemma}
	\label{lemma:laplacian_mult}
	    For every weighted graph $G = (V,E,w)$ and $v \in \mathbb{R}^{n}$, where $n=|V|$, it holds that $v^t \mathcal{L}(G) v = \sum_{(s,t)\in E}{w(s,t)\cdot (v_s - v_t)^2}$.
	\end{lemma}
    
    \begin{lemma}
    \label{lemma:matching_laplacian}
        Let $G = (V=[n],E,w)$ be a weighted undirected graph. For every $A \in \RR^{n\times n}$, it holds that $\tr(A' \mathcal{L}(G) A) =\sum_{\{i,j\}\in E}{w(i,j) \cdot \norm{\A{}{i}-\A{}{j}}_2^2}$, where $\A{}{i}\in \RR^n$ is row $i$ of $A$.
	\end{lemma}
    \begin{proof}
	    The $i$'th column of $A$ is $A'(i)\in\RR^n$. By Lemma~\ref{lemma:laplacian_mult}
	    \begin{align*}
	        \tr(A' \mathcal{L}(G) A) &= \sum_{k=1}^{n}{\left(A'(k)\right)' \mathcal{L}(G) A'(k)} 
    	    = \sum_{k=1}^{n}{\sum_{\{i,j\}\in E}{w(i,j)\cdot(\A{}{i,k}-\A{}{j,k})^2}}\\
    	    &=  \sum_{\{i,j\}\in E}{\sum_{k=1}^{n}{w(i,j)\cdot(\A{}{i,k}-\A{}{j,k})^2}}
    	    = \sum_{\{i,j\}\in E}{w(i,j)\cdot\norm{\A{}{i}-\A{}{j}}_2^2}
	    \end{align*}
	\end{proof}
    \begin{lemma}
    \label{lemma:normalized-laplacian-eigenvalues}
    	Let $A \in \RR^{n \times n}$ be a symmetric $\dGG$-stochastic $\ter$-blocked matrix, then the eigenvalues of the normalized Laplacian $\mathcal{N}(A) =I_{\ter} - \Dminushalf A \Dminushalf$ are in the range $[0,2]$. 
	\end{lemma}
    \begin{proof}
	    Since $\mathcal{L}(A) \ge 0$ we get that $\mathcal{N}(A) = \Dminushalf \mathcal{L}(A) \Dminushalf \ge 0$. Therefore, it is left to show that the eigenvalues of $\Dminushalf A \Dminushalf$ are at least $-1$. Let $X = U + A$. Note that $X \ge 0$. Indeed, similarly to Lemma~\ref{lemma:laplacian_mult} ,for any $v\in \RR^n$
	    \[
            v'Xv =
            \frac12\sum_{i=1}^n\sum_{j=1}^n
            \A{}{i,j}(v_i+v_j)^2 \ge 0.
        \]
	    Hence, $I_{\ter} + \Dminushalf A \Dminushalf = \Dminushalf X \Dminushalf \ge 0$ which means that the eigenvalues of $\Dminushalf A \Dminushalf$ are at least $-1$.
	\end{proof}
    \begin{lemma}
    \label{lemma:cs_median}
        Let $\{v_i\}_{i=1}^k$ be a set of $k$ vectors in $\RR^n$. Then,
        \begin{align*}
            k\norm{\frac{\sum_{i=1}^k{v_i}}{k}}_2^2 \le \sum_{i=1}^k{\norm{v_i}_2^2}
        \end{align*}
    \end{lemma}
    \begin{proof}
        By the Cauchy-Schwartz inequality, for all $j\in [n]$, 
        \begin{align*}
            k\left({\frac{\sum_{i=1}^k{v_{i,j}}}{k}}\right)^2 = \frac{1}{k}\left({\sum_{i=1}^k{v_{i,j}}}\right)^2 \le \frac{1}{k}\cdot k\cdot \sum_{i=1}^k{v_{i,j}^2} = \sum_{i=1}^k{v_{i,j}^2}
        \end{align*}
        Therefore,
        \begin{align*}
            k\norm{\frac{\sum_{i=1}^k{v_i}}{k}}_2^2 = \sum_{j=1}^n{k\left({\frac{\sum_{i=1}^k{v_{i,j}}}{k}}\right)^2} \le \sum_{j=1}^n{\sum_{i=1}^k{v_{i,j}^2}} = \sum_{i=1}^k{\sum_{j=1}^n{v_{i,j}^2}} = \sum_{i=1}^k{\norm{v_i}_2^2}
        \end{align*}
    \end{proof}
    
\section{Projection Lemmas}
\label{appendix:projection}
	The following fact is quite standard in the analysis of algorithms in the cut-matching framework. 
    \begin{lemma} [Gaussian Behavior of Projections; Variation of~\texorpdfstring{\cite[Lemma 3.5]{khandekar2009graph}}{[KRV09, Lemma 3.5]}]
    \label{lemma:projection_simple}
	    Let $\{v_i\}_{i=1}^k$ be a set of $k\le n$ vectors in $\RR^n$. For $i\in [k]$, let $u_i=\langle v_i,r\rangle$ be the projection of $v_i$ onto a random unit vector $r\in \mathbb{S}^{n-1}\subseteq \RR^n$. Then:
	    \begin{enumerate}
	        \item $\EE[u_i^2]=\frac{1}{n}\norm{v_i}_2^2$ for all $i$.
	        \item 
	        \[
    		        u_i^2 \le \frac{\alpha\log n}{n}\norm{v_i}_2^2
            \]
    	    holds for all $i$ with probability of at least $1-\frac{1}{n^{\alpha/8}}$, for any $\alpha\ge8$.
	    \end{enumerate}
	\end{lemma}
	\begin{remark}\label{remark:projection}
	    Lemma~\ref{lemma:projection_simple} follows from~\cite[Lemma 3.5]{khandekar2009graph} as follows: For a single projection, condition (2) fails to hold with probability at most $e^{-\alpha \log n/4}=n^{-\alpha/4}$. By the union bound it fails for some vector with probability $n^{1-\alpha/4}$. For this failure probability to be small we need $\alpha$ to be sufficiently large; If $\alpha \ge 8$ then the failure probability is at most $n^{-\alpha/8}$.
	\end{remark}

    \begin{corollary}
	\label{lemma:projection_pairs}
	    Let $\{v_i\}_{i=1}^k$ be a set of $k\le n$ vectors in $\RR^n$. For $i\in [k]$, let $u_i=\langle v_i,r\rangle$ be the projection of $v_i$ onto a random unit vector $r\in \mathbb{S}^{n-1}\subseteq \RR^n$. Then:
	    \begin{enumerate}
	        \item $\EE[u_i^2]=\frac{1}{n}\norm{v_i}_2^2$ for all $i$, and $\EE[(u_i - u_j)^2]=\frac{1}{n}\norm{v_i - v_j}_2^2$ for all pairs $(i,j)$.
	        \item 
	        \begin{align*}
	            u_i^2 &\le \frac{\alpha\log n}{n}\norm{v_i}_2^2
    		    \\
    		    (u_i - u_j)^2 &\le \frac{\alpha\log n}{n}\norm{v_i - v_j}_2^2
	        \end{align*}
    	    holds for all indices $i$ and pairs $(i,j)$ with probability of at least $1-\frac{1}{n^{\alpha/8}}$, for some constant $\alpha\ge16$.
	    \end{enumerate}
	\end{corollary}
	\begin{proof}
	    Apply Lemma~\ref{lemma:projection_simple} and Remark~\ref{remark:projection} on the set of vectors
	    \begin{align*}
	        \{v_i : i\in [k]\}\cup \{v_i - v_j : i,j\in[k]\}
	    \end{align*}
	    
	\end{proof}

    \section{Removing Assumption~\ref{assumption:mu}}
    \label{app:remove_assumption}

   In this appendix, we remove Assumption~\ref{assumption:mu} from the proof of Theorem~\ref{theorem:main-result3} and make the dependence on the vertex measure explicit. We bound the error terms in the projection argument in terms of the support size and the ratio of the largest to the smallest positive weight. The relevant $\log n$ factors are then replaced by $\log(k\mu_\Delta)$, where $k\defeq|\ter|$ and
    \[
    \mu_\Delta \defeq \frac{\mu_{\max}}{\mu_{\min}}, \qquad
    \mu_{\min} \defeq \min_{v\in \ter}\mu(v), \qquad
    \mu_{\max} \defeq \max_{v\in \ter}\mu(v).
    \]

    This replacement gives success probability with respect to $k\mu_\Delta$, rather than with respect to $n$. To recover high probability in $n$, we add a small algorithmic verification step: after each run, we estimate the final potential using a Gaussian trace estimator (see Corollary~\ref{cor:approx_potential}), and we repeat the algorithm $O(\log n)$ times until the potential is verified to be small. Since a successful run is recognizable by this verifier, the final algorithm succeeds with high probability in $n$ while preserving near-linear running time, up to the additional polylogarithmic factors.

    \begin{theorem}[Refined Statement of Theorem~\ref{theorem:main-result3}]
    \label{theorem:balanced_cut_mu_modified}
        Given a graph $G=(V,E)$ of $n$ vertices and $m$ edges, a parameter $\phi > 0$, and a weighting $\mu : V \to \RR_{\ge 0}$ of the vertices, there exists a randomized algorithm which takes $\tilde{O}\left(m\cdot \text{polylog}(k,\mu_{\Delta})\right)$ time and must end in one of the following three cases:
        \begin{enumerate}
            \item We certify that $G$ has $\mu$-expansion $\Phi^{\mu}(G)=\Omega(\phi)$, with high probability. 
            \item We find a cut $(A,\bar{A})$ in $G$ of $\mu$-expansion $\Phi^{\mu}_G(A,\bar{A})=O(\phi\log(\mu_\Delta\cdot k))$, and $\mu(A), \mu(\bar{A})$ are both $\Omega\left(\frac{\mu(V)}{\log (\mu_\Delta\cdot k)}\right)$, \ie, we find a relatively balanced $\mu$-sparse cut. 
            \item We find a cut $(A,\bar{A})$ with $0<\mu(\bar A)\le \frac{\mu(V)}{2}$, $\Phi^{\mu}_G(A,\bar{A})=O(\phi\log (\mu_\Delta\cdot k))$, and with high probability, $G[A]$ is an $\Omega(\phi)$-expander with respect to $\mu$. 
        \end{enumerate}
    \end{theorem}

    As a first step towards proving Theorem~\ref{theorem:balanced_cut_mu_modified}, we normalize the measure so that $\mu_{\min}=1$. This is without loss of generality. Indeed, for any $\alpha>0$ and any cut $(S,\bar S)$ with $\min\{\mu(S),\mu(\bar S)\}>0$, we have
    \[
    \Phi^{\alpha\mu}(S,\bar S)
    =
    \frac{|E(S,\bar S)|}{\min\{\alpha\mu(S),\alpha\mu(\bar S)\}}
    =
    \frac{1}{\alpha}\Phi^\mu(S,\bar S).
    \]
    Thus replacing $\mu$ by $\alpha\mu$ simply scales all $\mu$-expansion values by a factor $1/\alpha$. By scaling the parameter $\phi$ by the same factor, all statements of Theorem~\ref{theorem:balanced_cut_mu_modified} are preserved. We therefore assume from now on that $\mu_{\min}=1$, and hence $\mu_{\max}=\mu_\Delta$ and $\mu(V)\le k\mu_\Delta$.

    Let $L \defeq \max\{2,k\mu_{\Delta}\}$.
    We now describe the changes required in the proof of Theorem~\ref{theorem:balanced_cut_mu_modified}. Throughout the cut-matching game, every relevant occurrence of $\log n$ is replaced by $\log L$. In particular, the cut player uses lazy-walk parameter $\delta=\Theta(\log L)$ and the game is run for $T=\Theta(\log^2 L)$ rounds. The matching player uses auxiliary edge capacities $c=\Theta\left(\frac{1}{\phi\log L}\right)$ in the corresponding flow problems.

    With these parameters, the same embedding argument gives cuts of $\mu$-expansion $O(\phi\log L)$, and the balance guarantee becomes $\Omega(\mu(V)/\log L)$.

    \subsection{The potential analysis}

    The main analytical change concerns Lemma~\ref{lemma:potential_step_2}, where Assumption~\ref{assumption:mu} was used to bound the error terms in the random projection argument. To obtain bounds in terms of $L$, we also modify the choice of the projection vector. Since every row of $\W{t}$ is supported on $\ter$, these rows lie in a $k$-dimensional coordinate subspace of $\RR^n$. In each round, independently of the preceding rounds, we sample $r$ uniformly from the unit sphere of this subspace and set
    \[
        u_i=\left\langle\frac{\W{t}{i}}{\sqrt{\muG{}{i}}},r\right\rangle
        \qquad\text{for }i\in A_t\cap\ter.
    \]
    Thus the projection identities have normalization factor $k$ rather than $n$. With this change, we obtain the following analogue of Lemma~\ref{lemma:potential_step_2}.

    \begin{lemma}
    \label{lemma:potential_step_2_generalized}
        For each round $t$ and every sufficiently large constant $\alpha$,
        \[
        \EE\left[\frac{1}{3}\sum_{\{i,j\}\in\M{t}}w_{ij}
        \norm{\frac{\W{t}{i}}{\sqrt{\dG{}{i}}}-\frac{\W{t}{j}}{\sqrt{\dG{}{j}}}}_2^2
        +\sum_{j\in S_t\cap\ter}\dG{}{j}
        \norm{\frac{\W{t}{j}}{\sqrt{\dG{}{j}}}}_2^2\right]
        \ge\frac{1}{3000\alpha\log L}\psi(t)-\frac{1}{L^{\alpha/16}},
        \]
        where the expectation is over the random unit vector $r$ supported on $\ter$ defined above.
    \end{lemma}

    \begin{proof}
        Fix the choices made before round $t$, and write $x_i=\W{t}{i}/\sqrt{\muG{}{i}}$ for $i\in A_t\cap\ter$. The projection identities in dimension $k$ give
        \[
            \EE[u_i^2]=\frac{\norm{x_i}_2^2}{k}.
        \]
        Applying the single-vector bound used in Remark~\ref{remark:projection} in dimension $k$, with threshold $\alpha\log L$, and taking a union bound over the $O(k^2)$ vectors $x_i$ and $x_i-x_j$, we obtain, for sufficiently large $\alpha$ and using $k\le L$,
        \begin{align*}
            \norm{x_i}_2^2&\ge\frac{k}{\alpha\log L}u_i^2,
            &&\text{for all }i\in A_t\cap\ter,\\
            \norm{x_i-x_j}_2^2&\ge\frac{k}{\alpha\log L}(u_i-u_j)^2,
            &&\text{for all }i,j\in A_t\cap\ter,
        \end{align*}
        simultaneously with probability at least $1-L^{-\alpha/8}$.

        Define $z$ to be the maximum of zero and the violations of these inequalities, as in the proof of Lemma~\ref{lemma:potential_step_2}. Then $z=0$ on the event above. Since $\mu_{\min}=1$ and $\norm{\W{t}}_2\le1$, we have $\norm{x_i}_2\le1$ and therefore $|u_i|\le1$ and $|u_i-u_j|\le2$. Consequently,
        \[
            0\le z\le\frac{4k}{\alpha\log L}.
        \]
        With $K=\muG{}{V}\le L$, it follows that
        \[
            2K\EE[z]
            \le\frac{8Lk}{\alpha\log L}L^{-\alpha/8}
            \le\frac{8}{\alpha\log L}L^{2-\alpha/8}
            \le L^{-\alpha/16}
        \]
        for sufficiently large $\alpha$.

        The calculation in the proof of Lemma~\ref{lemma:potential_step_2}, with $n/(\alpha\log n)$ replaced by $k/(\alpha\log L)$, now bounds the left-hand side of the claimed inequality from below by
        \begin{align*}
            \frac{k}{3000\alpha\log L}
            \sum_{i\in A_t\cap\ter}\muG{}{i}\EE[u_i^2]-2K\EE[z]
            &=\frac{\psi(t)}{3000\alpha\log L}-2K\EE[z]\\
            &\ge\frac{\psi(t)}{3000\alpha\log L}-L^{-\alpha/16},
        \end{align*}
        as required.
    \end{proof}

    As in the original proof, summing the potential decrease over $T=\Theta(\log^2 L)$ rounds gives the following analogue of Corollary~\ref{cor:total_potential}.

    \begin{corollary} 
    \label{cor:total_potential_generalized}
        With high probability with respect to $L$, over the choices of $r$, $\psi(T) \le \frac{1}{L}$.
    \end{corollary} 

    The issue is that Corollary~\ref{cor:total_potential_generalized} gives high probability with respect to $L$, not with respect to $n$. We address this algorithmically. The key observation is that the condition $\psi(T)\le O(1/L)$ can be verified efficiently with high probability in $n$ by approximating the trace $\tr(\W{T}^2)$, see Corollary~\ref{cor:approx_potential}. We then repeat the modified cut-matching game $O(\log n)$ times and accept a run only if its final potential is verified to be sufficiently small, achieving high probability in $n$.

    \begin{theorem}[Gaussian trace estimator, \cite{avron2011randomized}]\label{thm:approx-trace}
    Let $A \in \RR^{n \times n}$ be a symmetric positive semidefinite matrix. Let $\epsilon,\delta \in (0,1)$, and let
    $q \ge 20\epsilon^{-2}\log(2/\delta).$
    Sample independent Gaussian vectors $g_1,\ldots,g_q \sim N(0,I_n)$, and define
    \[
    \widehat{\tr}(A)
    \defeq
    \frac{1}{q}\sum_{j=1}^q g_j' A g_j .
    \]
    Then, with probability at least $1-\delta$,
    $
    \widehat{\tr}(A) \in 
    [(1-\epsilon)\tr(A),
    (1+\epsilon)\tr(A)].
    $
    \end{theorem}

    Recall from Section~\ref{section:main-result_theorem_proof} that vector products with $W_T$ can be computed in $\tilde{O}(mT\log L)$ time. Applying Theorem~\ref{thm:approx-trace} with $A = W_T^2$, $\epsilon = \frac{1}{2}$, and $\delta = \frac{1}{n^a}$ gives the following.

    \begin{corollary}[Approximating the potential]\label{cor:approx_potential}
    Let $\psi(T) \defeq \tr(\W{T}^2)$. For every constant $a \ge 1$, there is a randomized algorithm that computes, in $\tilde{O}(mT\log L)$ time, an estimate $\widehat \psi(T)$ such that, with probability at least $1-n^{-a}$,
    $\frac{1}{2}\psi(T) \le \widehat{\psi}(T) \le 2\psi(T)$. 
    \end{corollary}

\end{document}